\documentclass[12pt,cls,onecolumn]{IEEEtran}
\usepackage{subfigure,cite,graphicx,psfig,amsmath,amssymb,mathrsfs,epsf}

\begin{document}

\title{On the Order Optimality of Large-scale Underwater Networks}
\author{\large Won-Yong Shin, \emph{Member}, \emph{IEEE}, Daniel E. Lucani, \emph{Member}, \emph{IEEE}, \\Muriel M{\'e}dard,
\emph{Fellow}, \emph{IEEE}, Milica Stojanovic,
\emph{Fellow}, \emph{IEEE},\\and Vahid Tarokh, \emph{Fellow}, \emph{IEEE} \\
\thanks{This work was supported in part by the National Science Foundation under grants No. 0520075, 0831728 and
CNS-0627021, and by the ONR MURI grant No. N00014-07-1-0738,
subcontract \# 060786 issued by BAE Systems National Security
Solutions, Inc., and supported by the Defense Advanced Research
Projects Agency (DARPA) and the Space and Naval Warfare System
Center (SPAWARSYSCEN), San Diego under Contract No. N66001-06-C-2020
(CBMANET), by the National Science Foundation under grant No.
0831728, and by the ONR grant No. N00014-09-1-0700. The material in
this paper was presented in part at the IEEE International Symposium
on Information Theory, Austin, TX, June 2010 and was submitted in
part to the 2011 IEEE International Symposium on Information
Theory.}
\thanks{W.-Y. Shin and V. Tarokh are with the School of Engineering and Applied Sciences, Harvard
University, Cambridge, MA 02138 USA (E-mail:\{wyshin,
vahid\}@seas.harvard.edu).}
\thanks{D. E. Lucani was with the Research Laboratory of Electronics, Massachusetts Institute of Technology,
Cambridge, MA 02139 USA. He is now with the Instituto de
Telecomunicacoes, Universidade do Porto, Porto 4200-465, Portugal
(E-mail: dlucani@fe.up.pt).}
\thanks{M. M{\'e}dard is with the Research Laboratory of Electronics, Massachusetts Institute of Technology,
Cambridge, MA 02139 USA (E-mail: medard@mit.edu).}
\thanks{M. Stojanovic is with the ECE Department, Northeastern University, Boston, MA 02115 USA (E-mail: millitsa@mit.edu).}
        } \maketitle


\markboth{Submitted to IEEE Transactions on Information Theory}
{Shin {\em et al.}: On the Order Optimality of Large-scale
Underwater Networks}


\newtheorem{definition}{Definition}
\newtheorem{theorem}{Theorem}
\newtheorem{lemma}{Lemma}
\newtheorem{example}{Example}
\newtheorem{corollary}{Corollary}
\newtheorem{proposition}{Proposition}
\newtheorem{conjecture}{Conjecture}
\newtheorem{remark}{Remark}

\def \diag{\operatornamewithlimits{diag}}
\def \min{\operatornamewithlimits{min}}
\def \max{\operatornamewithlimits{max}}
\def \log{\operatorname{log}}
\def \max{\operatorname{max}}
\def \rank{\operatorname{rank}}
\def \out{\operatorname{out}}
\def \exp{\operatorname{exp}}
\def \arg{\operatorname{arg}}
\def \E{\operatorname{E}}
\def \tr{\operatorname{tr}}
\def \SNR{\operatorname{SNR}}
\def \dB{\operatorname{dB}}
\def \ln{\operatorname{ln}}

\def \be {\begin{eqnarray}}
\def \ee {\end{eqnarray}}
\def \ben {\begin{eqnarray*}}
\def \een {\end{eqnarray*}}

\begin{abstract}
Capacity scaling laws are analyzed in an underwater acoustic network
with $n$ regularly located nodes on a square, in which both
bandwidth and received signal power can be limited significantly. A
narrow-band model is assumed where the carrier frequency is allowed
to scale as a function of $n$. In the network, we characterize an
attenuation parameter that depends on the frequency scaling as well
as the transmission distance. Cut-set upper bounds on the throughput
scaling are then derived in both extended and dense networks having
unit node density and unit area, respectively. It is first analyzed
that under extended networks, the upper bound is inversely
proportional to the attenuation parameter, thus resulting in a
highly power-limited network. Interestingly, it is seen that the
upper bound for extended networks is intrinsically related to the
attenuation parameter but not the spreading factor. On the other
hand, in dense networks, we show that there exists either a
bandwidth or power limitation, or both, according to the path-loss
attenuation regimes, thus yielding the upper bound that has three
fundamentally different operating regimes. Furthermore, we describe
an achievable scheme based on the simple nearest-neighbor multi-hop
(MH) transmission. We show that under extended networks, the MH
scheme is order-optimal for all the operating regimes. An
achievability result is also presented in dense networks, where the
operating regimes that guarantee the order optimality are
identified. It thus turns out that frequency scaling is instrumental
towards achieving the order optimality in the regimes. Finally,
these scaling results are extended to a random network realization.
As a result, vital information for fundamental limits of a variety
of underwater network scenarios is provided by showing capacity
scaling laws.
\end{abstract}

\begin{keywords}
Achievability, attenuation parameter, bandwidth, capacity scaling
law, carrier frequency, cut-set upper bound, dense network, extended
network, multi-hop (MH), operating regime, path-loss attenuation
regime, power-limited, underwater acoustic network.
\end{keywords}

\newpage


\section{Introduction}

Gupta and Kumar's pioneering work~\cite{GuptaKumar:00} characterized
the connection between the number $n$ of nodes and the sum
throughput in a large-scale wireless radio network. They showed that
the total throughput scales as $\Theta(\sqrt{n/\log n})$ when a
multi-hop (MH) routing strategy is used for $n$ source--destination
(S--D) pairs randomly distributed in a unit area.\footnote{We use
the following notation: i) $f(x)=O(g(x))$ means that there exist
constants $C$ and $c$ such that $f(x)\le Cg(x)$ for all $x>c$. ii)
$f(x)=o(g(x))$ means that
$\underset{x\rightarrow\infty}\lim\frac{f(x)}{g(x)}=0$. iii)
$f(x)=\Omega(g(x))$ if $g(x)=O(f(x))$. iv) $f(x)=\omega(g(x))$ if
$g(x)=o(f(x))$. v) $f(x)=\Theta(g(x))$ if $f(x)=O(g(x))$ and
$g(x)=O(f(x))$~\cite{Knuth:76}.} MH schemes are then further
developed and analyzed
in~\cite{GuptaKumar:03,FranceschettiDouseTseThiran:07,XueXieKumar:05,ElGamalMammenPrabhakarShah:06,ElGamalMammen:06,NebatCruzBhardwaj:09,
ShinChungLee:09}, while their throughput per S--D pair scales far
slower than $\Theta(1)$. Recent
results~\cite{OzgurLevequeTse:07,NiesenGuptaShah:09} have shown that
an almost linear throughput in the radio network, i.e.
$\Theta(n^{1-\epsilon})$ for an arbitrarily small $\epsilon>0$,
which is the best we can hope for, is achievable by using a
hierarchical cooperation (HC) strategy.\footnote{Note that the HC
scheme deals with a subtle issue around quantization, which is not
our main concern in this work.} Besides the schemes
in~\cite{OzgurLevequeTse:07,NiesenGuptaShah:09}, there have been
other studies to improve the throughput of wireless radio networks
up to a linear scaling in a variety of network scenarios by using
novel techniques such as networks with node
mobility~\cite{GrossglauserTse:02}, interference
alignment~\cite{CadambeJafar:07}, and infrastructure
support~\cite{ZemlianovVeciana:05}.

Together with the studies in terrestrial radio networks, the
interest in study of underwater networks has been growing, due to
recent advances in acoustic communication
technology~\cite{PartanKuroseLevine:06,Stojanovic:07,LucaniMedardStojanovicJSAC:08,LucaniStojanovicMedardOcean:08}.
In underwater acoustic communication systems, both bandwidth and
received signal power are severely limited owing to the exponential
(rather than polynomial) path-loss attenuation with long propagation
distance and the frequency-dependent attenuation. This is a main
feature that distinguishes underwater systems from wireless radio
links. Hence, the system throughput is affected by not only the
transmission distance but also the useful bandwidth. Based on these
characteristics, network coding
schemes~\cite{GuoXieCuiWang:06,LucaniMedardStojanovic:07,LucaniMedardStojanovicJSAC:08}
have been presented for underwater acoustic channels, while network
coding showed better performance than MH routing in terms of
reducing transmission power. MH networking has further been
investigated in other simple but realistic network conditions that
take into account the practical issues of coding and
delay~\cite{CarbonelliMitra:06,ZhangStojanovicMitra:10}.

One natural question is what are the fundamental capabilities of
underwater networks in supporting a multiplicity of nodes that wish
to communicate concurrently with each other, i.e., multiple S--D
pairs, over an acoustic channel. To answer this question, the
throughput scaling for underwater networks was first
studied~\cite{LucaniMedardStojanovic:08}, where $n$ nodes were
arbitrarily located in a planar disk of unit area, as
in~\cite{GuptaKumar:00}, and the carrier frequency was set to a
constant independent of $n$. That work showed an upper bound on the
throughput of each node, based on the physical
model~\cite{GuptaKumar:00}, which scales as
$n^{-1/\alpha}e^{-W_0(\Theta(n^{-1/\alpha}))}$, where $\alpha$
corresponds to the spreading factor of the underwater channel and
$W_0$ represents the branch zero of the Lambert W
function~\cite{Chapeau-BlondeauMonir:02}.\footnote{The Lambert W
function is defined to be the inverse of the function
$z=W(z)e^{W(z)}$ and the branch satisfying $W(z)\ge-1$ is denoted by
$W_0(z)$.} Since the spreading factor typically has values in the
range $1 \leq \alpha \leq 2$~\cite{LucaniMedardStojanovic:08}, the
throughput per node decreases almost as $O(n^{-1/\alpha})$ for large
enough $n$, which is considerably faster than the $\Theta(\sqrt{n})$
scaling characterized for wireless radio
settings~\cite{GuptaKumar:00}.

In this paper, capacity scaling laws for underwater networks are
analyzed in two fundamental but different networks: extended
networks~\cite{XieKumar:04,JovicicViswanathKulkarni:04,XueXieKumar:05,FranceschettiDouseTseThiran:07,OzgurLevequeTse:07}
of unit node density and dense
networks~\cite{GuptaKumar:00,OzgurLevequeTse:07,ElGamalMammenPrabhakarShah:06}
of unit area.\footnote{Since the two networks represent both extreme
network realizations, a realistic one would be in-between. In
wireless radio networks, the work in~\cite{OzgurJohariTseLeveque:10}
generalized the results of~\cite{OzgurLevequeTse:07} to the case
where the network area can scale polynomially with the number $n$ of
nodes. In underwater networks, we leave this issue for further
study.} Unlike the work in~\cite{LucaniMedardStojanovic:08}, the
{\em information-theoretic} notion of network capacity is adopted in
terms of characterizing the model for successful transmission.
Especially, we are interested in the case where the carrier
frequency scales as a certain function of $n$ in a narrow-band
model. Such an assumption leads to a significant change in the
scaling behavior owing to the attenuation characteristics. Recently,
the optimal capacity scaling of wireless radio networks has been
studied in~\cite{LeeChung:10,OzgurLevequeTse:10ITA} according to
operating regimes that are determined by the relationship between
the carrier frequency and the number $n$ of nodes. The frequency
scaling scenario of our study essentially follows the same arguments
as those in~\cite{LeeChung:10,OzgurLevequeTse:10ITA}. We aim to
study both information-theoretic upper bounds and achievable rate
scalings while allowing the frequency scaling with $n$. To the best
of our knowledge, such an attempt has never been done before in
underwater networks.

We explicitly characterize an {\em attenuation parameter} that
depends on the transmission distance and also on the carrier
frequency, and then identify fundamental operating regimes depending
on the parameter. For networks with $n$ regularly distributed nodes,
we derive upper bounds on the total throughput scaling using the
cut-set bound. In extended networks, our upper bound is based on the
characteristics of power-limited regimes shown
in~\cite{OzgurLevequeTse:07}. We show that the upper bound is
inversely proportional to the attenuation parameter. This leads to a
highly power-limited network for all the operating regimes, where
power consumption is important in determining performance.
Interestingly, it is seen that contrary to the case of wireless
radio networks, our upper bound heavily depends on the attenuation
parameter but not on the spreading factor (corresponding to the
path-loss exponent in wireless networks). On the other hand, in
dense networks, our upper bound basically follows the arguments,
similarly as in~\cite{OzgurJohariTseLeveque:10}: there exists either
a bandwidth or power limitation, or both, according to the operating
regimes (i.e., path-loss attenuation regimes). Specifically, the
network is bandwidth-limited as the path-loss attenuation is small.
This is because network performance on the total throughput is
roughly linear in the bandwidth. However, at the medium attenuation
regime, the network is both bandwidth- and power-limited since the
amount of available bandwidth and received signal power affects the
performance. Finally, the network becomes power-limited as the
attenuation parameter increases exponentially with respect to more
than $\sqrt{n}$, i.e., at the high attenuation regime. Hence, our
results indicate that the upper bound for dense networks has three
fundamentally different operating regimes according to the
attenuation parameter. In addition, to show constructively our
achievability result for extended networks, we describe the
conventional nearest-neighbor MH transmission~\cite{GuptaKumar:00},
and analyze its achievable throughput. We show that under extended
networks, the achievable rate scaling based on the MH routing
exactly matches the upper bound on the capacity scaling for all the
operating regimes. An achievability is also presented in dense
networks by utilizing the existing MH routing scheme with a slight
modification---we identify the operating regimes such that the
optimal capacity scaling is guaranteed. Therefore, this key result
indicates that frequency scaling is instrumental towards achieving
the order optimality in the regimes. Furthermore, a random network
scenario is discussed in this work. We show that under extended
random networks, the conventional MH-based achievable scheme is not
order-optimal for any operating regimes.

The rest of this paper is organized as follows.
Section~\ref{SEC:System} describes our system and channel models. In
Section~\ref{SEC:Upper}, the cut-set upper bounds on the throughput
are derived. In Section~\ref{SEC:Achievability}, the achievable
throughput scalings are analyzed. These results are extended to the
random network case in Section~\ref{SEC:Random}. Finally, Section
\ref{SEC:Conc} summarizes the paper with some concluding remarks.

Throughout this paper the superscript $H$, $[\cdot]_{ki}$, and
$\|\cdot\|_2$ denote the conjugate transpose, the $(k, i)$-th
element, and the largest singular value, respectively, of a matrix.
${\bf I}_n$ is the identity matrix of size $n\times n$, $\tr(\cdot)$
is the trace, $\det(\cdot)$ is the determinant, and $|\mathcal{X}|$
is the cardinality of the set $\mathcal{X}$. $\mathbb{C}$ is the
field of complex numbers and $E[\cdot]$ is the expectation. Unless
otherwise stated, all logarithms are assumed to be to the base 2.


\section{System and Channel Models} \label{SEC:System}

We consider a two-dimensional underwater network that consists of
$n$ nodes on a square such that two neighboring nodes are 1 and
$1/\sqrt{n}$ units of distance apart from each other in extended and
dense networks, respectively,
i.e., a regular
network~\cite{XieKumar:04,JovicicViswanathKulkarni:04}. We randomly
pick a matching of S--D pairs, so that each node is the destination
of exactly one source. Each node has an average transmit power
constraint $P$ (constant), and we assume that the channel state
information (CSI) is available at all receivers, but not at the
transmitters. It is assumed that each node transmits at a rate
$T(n)/n$, where $T(n)$ denotes the total throughput of the network.

Now let us turn to channel modeling. We assume frequency-flat
channel of bandwidth $W$ Hz around carrier frequency $f$, which
satisfies $f\gg W$, i.e., narrow-band model. This is a highly
simplified model, but nonetheless one that suffices to demonstrate
the fundamental mechanisms that govern capacity scaling. Assuming
that all the nodes have perfectly directional transmissions, we also
disregard multipath propagation, and simply focus on a line-of-sight
channel between each pair of nodes used
in~\cite{OzgurLevequeTse:07,NiesenGuptaShah:09,OzgurJohariTseLeveque:10}.
An underwater acoustic channel is characterized by an attenuation
that depends on both the distance $r_{ki}$ between nodes $i$ and $k$
($i, k\in\{1,\cdots,n\}$) and the signal frequency $f$, and is given
by
\begin{equation}   \label{EQ:af}
A(r_{ki},f)=c_0 r_{ki}^{\alpha}a(f)^{r_{ki}}
\end{equation}
for some constant $c_0>0$ independent of $n$, where $\alpha$ is the
spreading factor and $a(f)>1$ is the absorption
coefficient~\cite{Stojanovic:07}. For analytical tractability, we
assume that the spreading factor $\alpha$ does not change throughout
the network, i.e., that it is the same from short to long range
transmissions, as in wireless radio
networks~\cite{GuptaKumar:00,FranceschettiDouseTseThiran:07,OzgurLevequeTse:07}.
The spreading factor describes the geometry of propagation and is
typically $1\leq \alpha \leq 2$---its commonly used values are
$\alpha=1$ for cylindrical spreading, $\alpha=2$ for spherical
spreading, and $\alpha=1.5$ for the so-called practical spreading.
Note that existing models of wireless networks typically correspond
to the case for which $a(f)=1$ (or a positive constant independent
of $n$) and $\alpha>2$.\footnote{The counterpart of $\alpha$ in
wireless radio channels is the path-loss exponent.}

A common empirical model gives $a(f)$ in dB/km for $f$ in kHz
as~\cite{BerkhovskikhLysanov:82,Stojanovic:07}:
\begin{equation} \label{EQ:afeq}
10\log a(f)=a_0+a_1
f^2+a_2\frac{f^2}{b_1+f^2}+a_3\frac{f^2}{b_2+f^2},
\end{equation}
where $\{a_0,\cdots,a_3,b_1,b_2\}$ are some positive constants
independent of $n$. As stated earlier, we will allow the carrier
frequency $f$ to scale with the number $n$ of nodes. As a
consequence, a wider range of both $f$ and $n$ is covered, similarly
as
in~\cite{OzgurJohariTseLeveque:10,LeeChung:10,OzgurLevequeTse:10ITA}.
In particular, we consider the case where the frequency scales at
arbitrarily increasing rates relative to $n$, which enables us to
really capture the dependence on the frequency in
performance.\footnote{Otherwise, the attenuation parameter $a(f)$
scales as $\Theta(1)$ from (\ref{EQ:afeq}), which is not a matter of
interest in this work.} The absorption $a(f)$ is then an increasing
function of $f$ such that
\begin{equation} \label{EQ:afTheta}
a(f)=\Theta\left(e^{c_1 f^2}\right)
\end{equation}
with respect to $f$ for some constant $c_1>0$ independent of $n$.

The noise $n_i$ at node $i\in\{1,\cdots,n\}$ in an acoustic channel
can be modeled through four basic sources: turbulence, shipping,
waves, and thermal noise~\cite{Stojanovic:07}. We assume that $n_i$
is the circularly symmetric complex additive colored Gaussian noise
with zero mean and power spectral density (PSD) $N(f)$, and thus the
noise is frequency-dependent. The overall PSD of four sources decays
linearly on the logarithmic scale in the frequency region 100 Hz --
100 kHz, which is the operating region used by the majority of
acoustic systems, and thus is approximately given
by~\cite{Coates:89,Stojanovic:07}
\begin{equation} \label{EQ:Nfeq}
\log N(f)= a_4-a_5\log f
\end{equation}
for some positive constants $a_4$ and $a_5$ independent of
$n$.\footnote{Note that in our operating frequencies, $a_5=1.8$ is
commonly used for the above approximation~\cite{Stojanovic:07}.}
This means that $N(f)=O(1)$ since
\begin{equation} \label{EQ:NfTheta}
N(f)=\Theta\left(\frac{1}{f^{a_5}}\right)
\end{equation}
in terms of $f$ increasing with $n$. From (\ref{EQ:afTheta}) and
(\ref{EQ:NfTheta}), we may then have the following relationship
between the absorption $a(f)$ and the noise PSD $N(f)$:
\begin{equation} \label{EQ:Nfaf}
N(f)=\Theta\left(\frac{1}{(\log a(f))^{a_5/2}}\right).
\end{equation}

From the narrow-band assumption, the received signal $y_k$ at node
$k \in \{1,\cdots,n\}$ at a given time instance is given by
\begin{equation}
y_k=\sum_{i\in I}h_{ki}x_i+n_k, \label{EQ:signal}
\end{equation}
where
\begin{equation}
h_{ki}=\frac{e^{j\theta_{ki}}}{\sqrt{A(r_{ki},f)}} \label{EQ:hki}
\end{equation}
represents the complex channel between nodes $i$ and $k$,
$x_i\in\mathbb{C}$ is the signal transmitted by node $i$, and $I
\subset \{1,\cdots,n\}$ is the set of simultaneously transmitting
nodes. The random phases $\theta_{ki}$ are uniformly distributed
over $[0, 2\pi)$ and independent for different $i$, $k$, and time.
We thus assume a narrow-band time-varying channel, whose gain
changes to a new independent value for every symbol. Note that this
random phase model is based on a far-field
assumption~\cite{OzgurLevequeTse:07,NiesenGuptaShah:09,OzgurJohariTseLeveque:10},\footnote{In~\cite{FranceschettiMiglioreMinero:09},
instead of simply taking the far-field assumption, the physical
limit of wireless radio networks has been studied under certain
conditions on scattering elements. Further investigation is also
required to see whether this assumption is valid for underwater
networks of unit node density in the limit of large number $n$ of
nodes.} which is valid if the wavelength is sufficiently smaller
than the minimum node separation.

Based on the above channel characteristics, operating regimes of the
network are identified according to the following physical
parameters: the absorption $a(f)$ and the noise PSD $N(f)$ which are
exploited here by choosing the frequency $f$ based on the number $n$
of nodes. In other words, if the relationship between $f$ and $n$ is
specified, then $a(f)$ and $N(f)$ can be given by a certain scaling
function of $n$ from (\ref{EQ:afTheta}) and (\ref{EQ:NfTheta}),
respectively.


\section{Cut-set Upper Bound} \label{SEC:Upper}

To access the fundamental limits of an underwater network, new
cut-set upper bounds on the total throughput scaling are analyzed
from an information-theoretic perspective~\cite{CoverThomas:91}.
Consider a given cut $L$ dividing the network area into two equal
halves, as in~\cite{OzgurLevequeTse:07,OzgurJohariTseLeveque:10}
(see Figs.~\ref{FIG:cut1} and~\ref{FIG:cut_dense} for extended and
dense networks, respectively). Under the cut $L$, source nodes are
on the left, while all nodes on the right are destinations. In this
case, we have an $\Theta(n)\times \Theta(n)$ multiple-input
multiple-output (MIMO) channel between the two sets of nodes
separated by the cut.

\subsection{Extended Networks} \label{SEC:Upperex}

In this subsection, an upper bound based on the power transfer
argument~\cite{OzgurLevequeTse:07} is established for extended
networks, where the information flow for a given cut $L$ is
proportional to the total received signal power from source nodes.
Note, however, that the present problem is not equivalent to the
conventional extended network framework~\cite{OzgurLevequeTse:07}
due to quite different channel characteristics, and the main result
is shown here in a somewhat different way by providing a simpler
derivation than that of~\cite{OzgurLevequeTse:07}.

As illustrated in Fig.~\ref{FIG:cut1}, let $S_L$ and $D_L$ denote
the sets of sources and destinations, respectively, for the cut $L$
in an extended network. We then take into account an approach based
on the amount of power transferred across the network according to
different operating regimes, i.e., path-loss attenuation regimes. As
pointed out in~\cite{OzgurLevequeTse:07}, the information transfer
from $S_L$ to $D_L$ is highly power-limited since all the nodes in
the set $D_L$ are ill-connected to the left-half network in terms of
power. This implies that the information transfer is bounded by the
total received power transfer, rather than the cardinality of the
set $D_L$. For the cut $L$, the total throughput $T(n)$ for sources
on the left is bounded by the (ergodic) capacity of the MIMO channel
between $S_L$ and $D_L$ under time-varying channel assumption, and
thus is given by
\begin{eqnarray} \label{EQ:Tn12}
T(n)\le \underset{{\bf Q}_L\ge0}\max E\left[\log\det\left({\bf
I}_{n/2}+\frac{1}{N(f)}{\bf H}_L{\bf Q}_L{\bf
H}_L^{H}\right)\right],
\end{eqnarray}
where ${\bf H}_L$ is the matrix with entries $[{\bf
H}_L]_{ki}=h_{ki}$ for $i\in S_L, k\in D_L$, and ${\bf Q}_L\in
\mathbb{C}^{\Theta(n)\times\Theta(n)}$ is the positive semi-definite
input signal covariance matrix whose $k$-th diagonal element
satisfies $[{\bf Q}_L]_{kk}\le P$ for $k\in S_L$.

The relationship (\ref{EQ:Tn12}) will be further specified in
Theorem~\ref{TH:Cutset}. Before that, we first apply the techniques
of~\cite{JovicicViswanathKulkarni:04,VeeravalliLiangSayeed:05} to
obtain the total power transfer of the set $D_{L}$. These techniques
involve the design of the optimal input signal covariance matrix
${\bf Q}_L$ in terms of maximizing the upper bound~(\ref{EQ:Tn12})
on the capacity. If the matrix ${\bf H}_L$ has independent entries,
each $h_{ki}$ of which is a {\em proper} complex random
variable~\cite{NeeserMassey:93}, and has the same distribution as
$-h_{ki}$ for $i\in S_L, k\in D_L$, then the optimal ${\bf Q}_L$ is
diagonal, i.e., the maximum in (\ref{EQ:Tn12}) is attained with
$[\tilde{\bf Q}_L]_{kk}= P$ for $k\in S_L$, where $\tilde{\bf Q}_L$
is the diagonal matrix. We start from the following lemma.

\begin{lemma} \label{LEM:proper}
Each element $h_{ki}$ of the channel matrix ${\bf H}_L$ is a proper
complex random variable, where $i\in S_L, k\in D_L$.
\end{lemma}

The proof of this lemma is presented in Appendix~\ref{PF:proper}. It
is readily proved that $h_{ki}$ has the same distribution as
$-h_{ki}$ for all $i$ and $k$ since the random phases $\theta_{ki}$
are uniformly distributed over $[0, 2\pi)$. Thus, using the result
of Lemma~\ref{LEM:proper}, we obtain the following result.

\begin{lemma} \label{LEM:diag}
The optimal input signal covariance matrix ${\bf Q}_L$ that
maximizes the upper bound~(\ref{EQ:Tn12}) is unique and is given by
the diagonal $\tilde{\bf Q}_L$ with entries $[\tilde{\bf Q}_L]_{kk}=
P$ for $k\in S_L$.
\end{lemma}

We refer to Section III of~\cite{VeeravalliLiangSayeed:05} for the
detailed proof. From Lemma~\ref{LEM:diag}, the expression
(\ref{EQ:Tn12}) is then rewritten as
\begin{eqnarray}
T(n) \!\!\!\!\!\!\!\!&&\le E\left[\log\det\left({\bf
I}_{n/2}+\frac{1}{N(f)}{\bf H}_L\tilde{\bf Q}_L{\bf
H}_L^{H}\right)\right] \nonumber\\ &&=E\left[\log\det\left({\bf
I}_{n/2}+\frac{P}{N(f)}{\bf H}_L{\bf H}_L^{H}\right)\right] \nonumber\\
&&\le E\left[\underset{k\in D_L}\sum\log
\left(1+\frac{P}{N(f)}\underset{i\in
S_L}\sum\big|h_{ki}\big|^2\right)\right] \nonumber\\ &&
=\underset{k\in D_L}\sum\log \left(1+\frac{P}{N(f)}\underset{i\in
S_L}\sum\frac{1}{A(r_{ki},f)}\right) \nonumber\\ &&\le
\underset{k\in D_L}\sum \underset{i\in
S_L}\sum\frac{P}{A(r_{ki},f)N(f)}, \label{EQ:Tntemp}
\end{eqnarray}
where the second inequality is obtained by applying generalized
Hadamard's inequality~\cite{ConstantinescuScharf:98} as
in~\cite{JovicicViswanathKulkarni:04,OzgurLevequeTse:07}. The last
two steps come from (\ref{EQ:af}) and the fact that $\log (1+x)\le
x$ for any $x$, which is only tight as $x$ is small. Note that the
right-hand side (RHS) of (\ref{EQ:Tntemp}) represents the total
amount of received signal-to-noise ratio (SNR) from the set $S_L$ of
sources to the set $D_L$ of destinations for a given cut $L$. To
further compute (\ref{EQ:Tntemp}), we define the following parameter
\begin{equation}
P_L^{(k)}=\frac{P}{c_0}\underset{i\in S_L}\sum
r_{ki}^{-\alpha}a(f)^{-r_{ki}} \label{EQ:PL}
\end{equation}
for some constant $c_0>0$ independent of $n$, which corresponds to
the total power received from the signal sent by all the sources
$i\in S_L$ at node $k$ on the right (see (\ref{EQ:af}) and
(\ref{EQ:hki})). For convenience, we now index the node positions
such that the source and destination nodes under the cut $L$ are
located at positions $(-i_x+1,i_y)$ and $(k_x,k_y)$, respectively,
for $i_x, k_x=1,\cdots,\sqrt{n}/2$ and $i_y, k_y=1,\cdots,\sqrt{n}$.
The scaling result of $P_{L}^{(k)}$ defined in (\ref{EQ:PL}) can
then be derived as follows.

\begin{lemma} \label{LEM:PL}
In an extended network, the term $P_{L}^{(k)}$ in (\ref{EQ:PL}) is
given by
\begin{eqnarray} \label{EQ:PLupper}
P_{L}^{(k)}=O\left(k_x^{1-\alpha}a(f)^{-k_x}\right),
\end{eqnarray}
where $k_x$ represents the horizontal coordinate of node $k\in D_L$
for $k_x=1,\cdots,\sqrt{n}/2$.
\end{lemma}

The proof of this lemma is presented in Appendix~\ref{PF:PL}. We are
now ready to show the cut-set upper bound in extended networks.

\begin{theorem} \label{TH:Cutset}
For an underwater regular network of unit node density, the total
throughput $T(n)$ is upper-bounded by
\begin{eqnarray} \label{EQ:Tnupper}
T(n)\le \frac{c_2\sqrt{n}}{a(f)N(f)},
\end{eqnarray}
where $c_2>0$ is some constant independent of $n$.
\end{theorem}

\begin{proof}
From (\ref{EQ:af}) and (\ref{EQ:Tntemp})--(\ref{EQ:PLupper}), we
obtain the following upper bound on the total throughput $T(n)$:
\begin{eqnarray}
T(n)\!\!\!\!\!\!\!&&\le \frac{1}{N(f)}\underset{k\in D_L}\sum
P_{L}^{(k)} \nonumber\\ &&\le
\frac{1}{N(f)}\sum_{k_x=1}^{\sqrt{n}/2}\sum_{k_y=1}^{\sqrt{n}}
P_{L}^{(k)} \nonumber\\ &&\le
\frac{c_3P\sqrt{n}}{N(f)}\sum_{k_x=1}^{\sqrt{n}/2}\frac{1}{k_x^{\alpha-1}a(f)^{k_x}}
\nonumber\\ &&\le
\frac{c_3P\sqrt{n}}{N(f)}\sum_{k_x=1}^{\sqrt{n}/2}\frac{1}{a(f)^{k_x}}
\nonumber\\ &&\le \frac{c_3P\sqrt{n}}{N(f)}\frac{1}{a(f)-1} \nonumber\\
&&\le \frac{c_4P\sqrt{n}}{a(f)N(f)}, \nonumber
\end{eqnarray}
where $c_3$ and $c_4$ are some positive constants independent of
$n$, which is equal to (\ref{EQ:Tnupper}). This completes the proof
of the theorem.
\end{proof}

We remark that this upper bound is expressed as a function of the
absorption $a(f)$ and the noise PSD $N(f)$, whereas an upper bound
for wireless radio networks depends only on the constant value
$\alpha$~\cite{OzgurLevequeTse:07}.
In addition, using (\ref{EQ:afTheta}), (\ref{EQ:NfTheta}), and
(\ref{EQ:Nfaf}) in (\ref{EQ:Tnupper}) results in two other
expressions on the total throughput
\begin{equation} \label{EQ:Tnaf}
T(n)=O\left(\frac{\sqrt{n}\left(\log
a(f)\right)^{a_5/2}}{a(f)}\right)
\end{equation}
and
\begin{equation}
T(n)=O\left(\frac{\sqrt{n}f^{a_5}}{e^{c_1f^2}}\right) \nonumber
\end{equation}
for some positive constants $c_1$ and $a_5$ shown in
(\ref{EQ:afTheta}) and (\ref{EQ:Nfeq}), respectively. Hence, from
(\ref{EQ:Tnaf}), it is seen that the upper bound is inversely
proportional to the attenuation parameter $a(f)$ and decays fast
with increasing $a(f)$, thereby leading to a highly power-limited
network irrespective of the parameter $a(f)$.

%

\subsection{Dense Networks} \label{SEC:upperdense}

In a dense network, it is necessary to narrow down the class of S--D
pairs according to their Euclidean distance to obtain a tight upper
bound. In this subsection, we derive a new upper bound based on
hybrid approaches that consider either the sum of the capacities of
the multiple-input single-output (MISO) channel between transmitters
and each receiver or the amount of power transferred across the
network according to operating regimes, similarly as
in~\cite{OzgurJohariTseLeveque:10}.

For the cut $L$, the total throughput $T(n)$ for sources on the left
half is bounded by the capacity of the MIMO channel between $S_L$
and $D_L$, corresponding to the sets of sources and destinations,
respectively, and thus is given by (\ref{EQ:Tn12}). In the extended
network framework, upper bounding the capacity by the total received
SNR yields a tight bound due to poor power connections for all the
operating regimes. In a dense network, however, we may have
arbitrarily high received SNR for nodes in the set $D_L$ that are
located close to the cut, or even for all the nodes, depending on
the path-loss attenuation regimes, and thus need the separation
between destination nodes that are well- and ill-connected to the
left-half network in terms of power. More precisely, the set $D_L$
of destinations is partitioned into two groups $D_{L,1}$ and
$D_{L,2}$ according to their location, as illustrated in
Fig.~\ref{FIG:cut_dense}. Then, since Lemmas~\ref{LEM:proper}
and~\ref{LEM:diag} also hold for the dense network, by applying
generalized Hadamard's inequality~\cite{ConstantinescuScharf:98}, we
have
\begin{eqnarray}
T(n)\!\!\!\!\!\!\!&& \le \underset{{\bf Q}_L\ge0}\max
E\left[\log\det\left({\bf I}_{n/2}+\frac{1}{N(f)}{\bf H}_{L}{\bf
Q}_L{\bf H}_{L}^{H}\right)\right] \nonumber
\\ && \le
E\left[\log\det\left({\bf I}_{n/2}+\frac{P}{N(f)}{\bf H}_{L}{\bf
H}_{L}^{H}\right)\right] \nonumber
\\ && \le E\left[\log\det\left({\bf
I}_{|D_{L,1}|}+\frac{P}{N(f)}{\bf H}_{L,1}{\bf
H}_{L,1}^{H}\right)\right] \nonumber\\ && +
E\left[\log\det\left({\bf I}_{|D_{L,2}|}+\frac{P}{N(f)}{\bf
H}_{L,2}{\bf H}_{L,2}^{H}\right)\right], \label{EQ:TRHS}
\end{eqnarray}
where ${\bf H}_{L,l}$ is the matrix with entries $[{\bf
H}_{L,l}]_{ki}=h_{ki}$ for $i\in S_L$, $k\in D_{L,l}$, and $l=1,2$.
Note that the first and second terms in the RHS of (\ref{EQ:TRHS})
represent the information transfer from $S_L$ to $D_{L,1}$ and from
$S_L$ to $D_{L,2}$, respectively. Here, $D_{L,1}$ denotes the set of
destinations located on the rectangular slab of width $x_L/\sqrt{n}$
immediately to the right of the centerline (cut), where
$x_L\in\{0,1,\cdots,\sqrt{n}/2\}$. The set $D_{L,2}$ is given by
$D_L \setminus D_{L,1}$. It then follows that
$|D_{L,1}|=x_L\sqrt{n}$ and $|D_{L,2}|=(\sqrt{n}/2-x_L)\sqrt{n}$.

Let $T_l(n)$ denote the $l$-th term in the RHS of (\ref{EQ:TRHS})
for $l\in\{1,2\}$. It is then reasonable to bound $T_1(n)$ by the
cardinality of the set $D_{L,1}$ rather than the total received SNR.
In contrast, using the power transfer argument for the term
$T_2(n)$, as in the extended network case, will lead to a tight
upper bound. It is thus important to set the parameter $x_L$
according to the attenuation parameter $a(f)$, based on the
selection rule for $x_L$~\cite{OzgurJohariTseLeveque:10}, so that
only $D_{L,1}$ contains the destination nodes with high received
SNR. To be specific, we need to decide whether the SNR received by a
destination $k\in D_L$ from the set $S_L$ of sources, denoted by
$P_L^{(k)}/N(f)$, is larger than one, because degrees-of-freedom
(DoFs) (also known as capacity pre-log factor) of the MISO channel
are limited to one. If destination node $k$ has the total received
SNR greater than one, i.e., $P_L^{(k)}=\omega(N(f))$, then it
belongs to $D_{L,1}$. Otherwise, it follows that $k\in D_{L,2}$.

For analytical tractability, suppose that
\begin{equation}
a(f)=\Theta\left((1+\epsilon_0)^{n^\beta}\right) \textrm{    for
$\beta\in[0,\infty)$}, \label{EQ:afepsilon}
\end{equation}
where $\epsilon_0>0$ is an arbitrarily small constant, which
represents all the operating regimes with varying $\beta$. As
before, let us index the node positions such that the source and
destination nodes are located at positions
$\left(\frac{-i_x+1}{\sqrt{n}},\frac{i_y}{\sqrt{n}}\right)$ and
$\left(\frac{k_x}{\sqrt{n}},\frac{k_y}{\sqrt{n}}\right)$,
respectively, for $i_x, k_x=1,\cdots,\sqrt{n}/2$ and $i_y,
k_y=1,\cdots,\sqrt{n}$. We then obtain the following scaling results
for $P_L^{(k)}$ as shown below.

\begin{lemma} \label{LEM:PLdense}
In a dense network, the term $P_L^{(k)}$ in (\ref{EQ:PL}) is given
by
\begin{eqnarray}
P_L^{(k)}=\left\{\begin{array}{lll} O(n) &\textrm{if
$1\le\alpha<2$ and $k_x=o\left(n^{1/2-\beta+\epsilon}\right)$} \\
O\left(n\log n\right) &\textrm{if $\alpha=2$ and
$k_x=o\left(n^{1/2-\beta+\epsilon}\right)$}
\\
O\left(\frac{n^{\alpha/2}}{(1+\epsilon_0)^{k_x
n^{\beta-1/2}}}\max\left\{1,n^{1/2-\beta}\right\}\right) &\textrm{if
$k_x=\Omega\left(n^{1/2-\beta+\epsilon}\right)$}
\end{array}\right. \label{EQ:PLupperdense}
\end{eqnarray}
and
\begin{eqnarray}
P_L^{(k)}=\left\{\begin{array}{lll} \Omega\left(\frac{n^{\alpha/2-\epsilon}}{k_x^{\alpha-1}}\right) &\textrm{if $k_x=o\left(n^{1/2-\beta+\epsilon}\right)$} \\
\Omega\left(\frac{1}{\left(1+\epsilon_0\right)^{k_x
n^{\beta-1/2}}}\max\left\{1,\frac{n^{1/2-\beta}}{\left(1+\epsilon_0\right)^{n^{\beta-1/2}}}\right\}\right)
&\textrm{if $k_x=\Omega\left(n^{1/2-\beta+\epsilon}\right)$}
\end{array}\right. \label{EQ:PLlowerdense}
\end{eqnarray}
for arbitrarily small positive constants $\epsilon$ and
$\epsilon_0$, where $k_x/\sqrt{n}$ is the horizontal coordinate of
node $k\in D_{L,2}$.
\end{lemma}


The proof of this lemma is presented in Appendix~\ref{PF:PLdense}.
Although the upper and lower bounds for $P_L^{(k)}$ are not
identical to each other, showing these scaling results is sufficient
to make a decision on $x_L$ according to the operating regimes. It
is seen from Lemma~\ref{LEM:PLdense} that when
$k_x=o\left(n^{1/2-\beta+\epsilon}\right)$, $P_L^{(k)}$ does not
depend on the parameter $\beta$ (or $a(f)$), while for
$k_x=\Omega\left(n^{1/2-\beta+\epsilon}\right)$, node $k\in D_{L,2}$
gets ill-connected to the left half in terms of power since
$P_L^{(k)}$ decreases exponentially with $n$. More specifically,
when $k_x=o\left(n^{1/2-\beta+\epsilon}\right)$, it follows that
$P_L^{(k)}=\omega(n^{\alpha\beta})$ from (\ref{EQ:PLlowerdense}),
resulting in $P_L^{(k)}=\omega(N(f))$ due to $N(f)=O(1)$. In
contrast, under the condition
$k_x=\Omega\left(n^{1/2-\beta+\epsilon}\right)$, it is observed from
(\ref{EQ:PLupperdense}) that $P_L^{(k)}$ is exponentially decaying
as a function of $n$, thus leading to $P_L^{(k)}=o(N(f))$. As a
consequence, using the result of Lemma~\ref{LEM:PLdense}, three
different regimes are identified and the selected $x_L$ is specified
accordingly:
\begin{eqnarray} \label{EQ:xLdecision}
x_L=\left\{\begin{array}{lll} \sqrt{n}/2 &\textrm{if $\beta=0$}
\\ n^{1/2-\beta+\epsilon} &\textrm{if
$0<\beta\le1/2$} \\ 0 &\textrm{if $\beta>1/2$}
\end{array}\right.
\end{eqnarray}
for an arbitrarily small $\epsilon>0$. It is now possible to show
the proposed cut-set upper bound in dense networks.

\begin{theorem} \label{THM:upperdense}
Consider an underwater regular network of unit area. Then, the upper
bound on the total throughput $T(n)$ is given by
\begin{eqnarray} \label{EQ:Tndensefinal}
T(n)=\left\{\begin{array}{lll} O(n\log n) &\textrm{if $\beta=0$}
\\ O\left(n^{1-\beta+\epsilon}\log n\right) &\textrm{if
$0<\beta\le1/2$} \\
O\left(\frac{n^{(1+\alpha+\beta
a_5)/2}}{(1+\epsilon_0)^{n^{\beta-1/2}}}\right) &\textrm{if
$\beta>1/2$,}
\end{array}\right.
\end{eqnarray}
where $\epsilon$ and $\epsilon_0$ are arbitrarily small positive
constants, and $a_5$ and $\beta$ are defined in (\ref{EQ:Nfeq}) and
(\ref{EQ:afepsilon}), respectively.
\end{theorem}

\begin{proof}
We first compute the first term $T_1(n)$ in (\ref{EQ:TRHS}),
focusing on the case for $0\le\beta\le1/2$ since otherwise
$T_1(n)=0$. Since the nodes in the set $D_{L,1}$ have good power
connections to the left-half network and the information transfer to
$D_{L,1}$ is limited in bandwidth (but not power), the term $T_1(n)$
is upper-bounded by the sum of the capacities of the MISO channels.
More specifically, by generalized Hadamard's
inequality~\cite{ConstantinescuScharf:98}, $T_1(n)$ can be easily
bounded by
\begin{eqnarray}
T_1(n)\!\!\!\!\!\!\!&& \le \underset{k\in D_{L,1}}\sum\log
\left(1+\frac{P}{N(f)}\underset{i\in S_L}\sum
\frac{1}{A(r_{ki},f)}\right) \nonumber\\ &&\le x_L \sqrt{n}
\log\left(1+\frac{Pn^{\alpha/2+1}}{a(f)^{1/\sqrt{n}}N(f)}\right)
\nonumber\\ &&\le x_L \sqrt{n}
\log\left(1+\frac{Pn^{\alpha/2+1}}{N(f)}\right) \nonumber\\ &&\le
c_5 x_L \sqrt{n}\log n \label{EQ:denseT1}
\end{eqnarray}
for some constant $c_5>0$ independent of $n$, where the last two
steps are obtained from the fact that $0<a(f)\le1$ and $N(f)$ tends
to decrease polynomially with $n$ from the relation in
(\ref{EQ:Nfaf}). The upper bound for the second term $T_2(n)$ in
(\ref{EQ:TRHS}) is now derived under the condition
$\beta\in(0,\infty)$. Similarly as in the steps of
(\ref{EQ:Tntemp}), we have
\begin{eqnarray}
T_2(n)\!\!\!\!\!\!\!&& \le \underset{k\in D_{L,2}}\sum\log
\left(1+\frac{P}{N(f)}\underset{i\in S_L}\sum
\frac{1}{A(r_{ki},f)}\right) \nonumber\\ && \le \underset{k\in
D_{L,2}}\sum\underset{i\in S_L}\sum \frac{P}{A(r_{ki},f)N(f)}
\nonumber\\ && =\frac{1}{N(f)}\underset{k\in D_{L,2}}\sum P_L^{(k)},
\label{EQ:denseT2}
\end{eqnarray}
which corresponds to the sum of the total received SNR from the
left-half network to the destination set $D_{L,2}$. Hence, combining
the two bounds (\ref{EQ:denseT1}) and (\ref{EQ:denseT2}) along with
the choices for $x_L$ specified in (\ref{EQ:xLdecision}), we obtain
the following upper bound on the total throughput $T(n)$:
\begin{eqnarray}
T(n)\!\!\!\!\!\!\!&& \le \left\{\begin{array}{lll} c_5n\log n
&\textrm{if
$\beta=0$} \\
c_5n^{1-\beta+\epsilon}\log n+\frac{1}{N(f)}\underset{k\in
D_{L,2}}\sum P_L^{(k)} &\textrm{if
$0<\beta\le1/2$} \\
\frac{1}{N(f)}\underset{k\in D_{L}}\sum P_L^{(k)} &\textrm{if
$\beta>1/2$}
\end{array}\right. \nonumber\\ && = \left\{\begin{array}{lll} c_5n\log n &\textrm{if
$\beta=0$} \\
c_5 n^{1-\beta+\epsilon}\log n+\frac{1}{N(f)}
\sum_{k_x=x_L}^{\sqrt{n}/2}\sum_{k_y=1}^{\sqrt{n}}\frac{n^{\alpha/2+1/2-\beta}}{(1+\epsilon_0)^{k_x
n^{\beta-1/2}}} &\textrm{if
$0<\beta\le1/2$} \\
\frac{1}{N(f)}\sum_{k_x=1}^{\sqrt{n}/2}\sum_{k_y=1}^{\sqrt{n}}
\frac{n^{\alpha/2}}{(1+\epsilon_0)^{k_x n^{\beta-1/2}}} &\textrm{if
$\beta>1/2$}
\end{array}\right. \nonumber\\ && \le \left\{\begin{array}{lll} c_5n\log
n &\textrm{if
$\beta=0$} \\
c_6n^{1-\beta+\epsilon}\log n &\textrm{if
$0<\beta\le1/2$} \\
\frac{n^{(1+\alpha)/2}}{N(f)}\sum_{k_x=1}^{\sqrt{n}/2}
\frac{1}{(1+\epsilon_0)^{k_x n^{\beta-1/2}}} &\textrm{if
$\beta>1/2$}
\end{array}\right. \label{EQ:Tnregimes}
\end{eqnarray}
for an arbitrarily small $\epsilon>0$ and some constant $c_6>0$
independent of $n$, where the equality comes from
(\ref{EQ:PLupperdense}). The second inequality holds due to the fact
that the term $\frac{n^{\alpha/2}}{(1+\epsilon_0)^{k_x
n^{\beta-1/2}}}$ tends to decay exponentially with $n$ under the
conditions $0<\beta\le1/2$ and $x_L\le k_x\le\sqrt{n}/2$, and thus
the total is dominated by the information transfer $T_1(n)$ in
(\ref{EQ:denseT1}). Now let us focus on the last line of
(\ref{EQ:Tnregimes}), which corresponds to the total amount of SNR
received by all nodes for the condition $\beta>1/2$. For
$\beta>1/2$, using the two relationships (\ref{EQ:Nfaf}) and
(\ref{EQ:afepsilon}) follows that
\begin{eqnarray}
\frac{n^{(1+\alpha)/2}}{N(f)}\sum_{k_x=1}^{\sqrt{n}/2}
\frac{1}{(1+\epsilon_0)^{k_x n^{\beta-1/2}}} \!\!\!\!\!\!\!&& \le
\frac{n^{(1+\alpha)/2}}{N(f)}
\frac{1}{(1+\epsilon_0)^{n^{\beta-1/2}}-1} \nonumber\\ && \le
\frac{c_7 n^{(1+\alpha+\beta
a_5)/2}}{(1+\epsilon_0)^{n^{\beta-1/2}}} \nonumber
\end{eqnarray}
for some constant $c_7>0$ independent of $n$, where the second
inequality holds due to $(1+\epsilon_0)^{n^{\beta-1/2}}=\omega(1)$
under the condition. This coincides with the result shown in
(\ref{EQ:Tndensefinal}), which completes the proof.
\end{proof}

Note that the upper bound~\cite{OzgurLevequeTse:07} for wireless
radio networks of unit area is given by $O(n \log n)$, which is the
same as the case with $\beta=0$ (or equivalently $a(f)=\Theta(1)$)
in the dense underwater network. Now let us discuss the fundamental
limits of the network according to three different operating regimes
shown in (\ref{EQ:Tndensefinal}).

\begin{remark}
The upper bound on the total capacity scaling is illustrated in
Fig.~\ref{FIG:bounds} versus the parameter $\beta$ (logarithmic
terms are omitted for convenience). We first address the regime
$\beta=0$ (i.e., low path-loss attenuation regime), in which the
upper bound on $T(n)$ is active with $x_L=\sqrt{n}/2$, or
equivalently $D_{L,1}=D_L$, while $T_2(n)=0$. In this case, the
total throughput of the network is limited by the DoFs of the
$\Theta(n)\times \Theta(n)$ MIMO channel between $S_L$ and $D_L$,
and is roughly linear in the bandwidth, thus resulting in a
bandwidth-limited network. In particular, our interest is in the
operating regimes for which the network becomes power-limited as
$\beta>0$. In the second regime $0<\beta\le 1/2$ (i.e., medium
path-loss attenuation regime), as pointed out in the proof of
Theorem~\ref{THM:upperdense}, the upper bound on $T(n)$ is dominated
by the information transfer from $S_L$ to $D_{L,1}$, that is, the
term $T_1(n)$ in (\ref{EQ:denseT1}) contributes more than $T_2(n)$
in (\ref{EQ:denseT2}). The total throughput is thus limited by the
DoFs of the MIMO channel between $S_L$ and $D_{L,1}$, since more
available bandwidth leads to an increment in $T_1(n)$. As a
consequence, in this regime, the network is both bandwidth- and
power-limited. In the third regime $\beta>1/2$ (i.e., high path-loss
attenuation regime), the upper bound (\ref{EQ:denseT2}) is active
with $x_L=0$, or equivalently $D_{L,2}=D_L$, while $T_1(n)=0$. The
information transfer to $D_L$ is thus totally limited by the sum of
the total received SNR from the left-half network to the destination
nodes in $D_L$. In other words, in the third regime, the network is
limited in power, but not in bandwidth.
\end{remark}

Note that the upper bound on $T(n)$ decays polynomially with
increasing $\beta$ in the regime $0<\beta\le 1/2$, while it drops
off exponentially when $\beta>1/2$. In addition, two other
expressions on the total throughput $T(n)$ are summarized as
follows.

\begin{remark}
From (\ref{EQ:Nfaf}) and (\ref{EQ:afepsilon}), the upper bound and
the corresponding operating regimes can also be presented below in
terms of the attenuation parameter $a(f)$:
\begin{eqnarray}
T(n)=\left\{\begin{array}{lll} O(n\log n) &\textrm{if
$a(f)=\Theta(1)$}
\\ O\left(\frac{n^{1+\epsilon}\log n}{\log a(f)}\right) &\textrm{if
$a(f)=\omega(1)$ and $a(f)=O\left((1+\epsilon_0)^{\sqrt{n}}\right)$} \\
O\left(\frac{n^{(1+\alpha)/2} \left(\log
a(f)\right)^{a_5/2}}{a(f)^{1/\sqrt{n}}}\right) &\textrm{if
$a(f)=\omega\left((1+\epsilon_0)^{\sqrt{n}}\right)$.}
\end{array}\right. \nonumber
\end{eqnarray}
Note that as $a(f)=\omega\left((1+\epsilon_0)^{\sqrt{n}}\right)$, we
also obtain
\begin{equation}
T(n)=O\left(\frac{n^{(1+\alpha)/2}}{a(f)^{1/\sqrt{n}}N(f)}\right),
\nonumber
\end{equation}
which is expressed as a function of the spreading factor $\alpha$ as
well as the absorption $a(f)$ and the noise PSD $N(f)$. Using
(\ref{EQ:afTheta}) and (\ref{EQ:NfTheta}) further yields the
following expression
\begin{eqnarray}
T(n)=\left\{\begin{array}{lll} O(n\log n) &\textrm{if $f=\Theta(1)$}
\\ O\left(\frac{n^{1+\epsilon}\log n}{f^2}\right) &\textrm{if
$f=\omega(1)$ and $f=O\left(n^{1/4}\right)$} \\
O\left(\frac{n^{(1+\alpha)/2}f^{a_5}}{e^{c_1f^2/\sqrt{n}}}\right)
&\textrm{if $f=\omega\left(n^{1/4}\right)$,}
\end{array}\right. \nonumber
\end{eqnarray}
which represents the upper bound for the operating regimes
identified by frequency scaling.
\end{remark}


\section{Achievability Result} \label{SEC:Achievability}

In this section, to show the order optimality in underwater
networks, we analyze the achievable throughput scaling for both
extended and dense networks with the existing transmission scheme,
commonly used in wireless radio networks. Under an extended regular
network, the conventional MH transmission is introduced and its
optimal achievability result is shown. Under a dense regular
network, we examine the operating regimes for which the achievable
throughput matches the upper bound shown in
Section~\ref{SEC:upperdense}.

\subsection{Extended Networks}

The nearest-neighbor MH routing protocol~\cite{GuptaKumar:00} will
be briefly described to show the order optimality. The basic
procedure of the MH protocol under our extended regular network is
as follows:
\begin{itemize}
\item Divide the network into square routing cells, each of which
has unit area.
\item Draw an line connecting a S--D pair. A source transmits a packet to its destination using the nodes in the adjacent cells passing through the line.
\item Use the full transmit power at each node, i.e., the transmit power $P$.
\end{itemize}


The achievable rate of MH is now shown by quantifying the amount of
interference.

\begin{lemma}   \label{LEM:interference}
Consider an extended regular network that uses the nearest-neighbor
MH protocol. Then, the total interference power $P_I$ from other
simultaneously transmitting nodes, corresponding to the set
$I\subset \{1,\cdots,n\}$, is upper-bounded by $\Theta(1/a(f))$,
where $a(f)$ denotes the absorption coefficient greater than 1.
\end{lemma}

\begin{proof}
There are $8k$ interfering routing cells, each of which includes one
node, in the $k$-th layer $l_k$ of the network as illustrated in
Fig.~\ref{FIG:layer}. Then from (\ref{EQ:af}), (\ref{EQ:signal}),
and (\ref{EQ:hki}), the total interference power $P_I$ at each node
from simultaneously transmitting nodes is upper-bounded by
\begin{eqnarray}
P_I\!\!\!\!\!\!\!&&=\sum_{k=1}^{\infty}(8k)\frac{P}{c_0k^{\alpha}a(f)^{k}}
\nonumber\\&&=\frac{8P}{c_0}\sum_{k=1}^{\infty}\frac{1}{k^{\alpha-1}a(f)^{k}}
\nonumber\\ &&\le
\frac{8P}{c_0}\sum_{k=1}^{\infty}\frac{1}{a(f)^{k}} \nonumber\\
&&\le \frac{c_8}{a(f)}, \nonumber
\end{eqnarray}
where $c_0$ and $c_8$ are some positive constants independent of
$n$, which completes the proof.
\end{proof}

Note that the received signal power no longer decays polynomially
but rather exponentially with propagation distance in our network.
This implies that the absorption term $a(f)$ in (\ref{EQ:af}) will
play an important role in determining the performance. It is also
seen that the upper bound on $P_I$ does not depend on the spreading
factor $\alpha$. Using Lemma~\ref{LEM:interference}, it is now
possible to simply obtain a lower bound on the capacity scaling in
the network, and hence the following result presents the achievable
rates under the MH protocol.

\begin{theorem} \label{TH:achievable}
In an underwater regular network of unit node density,
\begin{equation}   \label{EQ:achievablerate}
T(n)=\Omega\left(\frac{n^{1/2}}{a(f)N(f)}\right)
\end{equation}
is achievable.
\end{theorem}

\begin{proof}
Suppose that the nearest-neighbor MH protocol is used. To get a
lower bound on the capacity scaling, the
signal-to-interference-and-noise ratio (SINR) seen by receiver
$i\in\{1,\cdots,n\}$ is computed as a function of the absorption
$a(f)$ and the PSD $N(f)$ of noise $n_i$. Since the Gaussian is the
worst additive noise~\cite{Medard:00,DiggaviCover:01}, assuming it
lower-bounds the throughput. Hence, by assuming full CSI at the
receiver, from (\ref{EQ:af}), (\ref{EQ:signal}), and (\ref{EQ:hki}),
the achievable throughput per S--D pair is lower-bounded by
\begin{eqnarray}
&& \log (1+\text{SINR}) \nonumber\\
\ge\!\!\!\!\!\!\!&&\log\left(1+\frac{P/(c_0a(f))}{N(f)+c_8/a(f)}\right)
 \nonumber\\
\ge\!\!\!\!\!\!\!&&\log \left(1+\frac{c_9P}{a(f)N(f)}\right),
\label{EQ:SINR} \nonumber
\end{eqnarray}
for some positive constants $c_0$, $c_8$, and $c_9$ independent of
$n$, where the second inequality is obtained from the relationship
(\ref{EQ:Nfaf}) between $a(f)$ and $N(f)$, resulting in
$N(f)=\Omega\left(1/a(f)\right)$. Due to the fact that
$\log(1+x)=(\log e)x+O(x^2)$ for small $x>0$, the rate of
\begin{equation}
\Omega\left(\frac{1}{a(f)N(f)}\right) \nonumber
\end{equation}
is thus provided for each S--D pair. Since the number of hops per
S--D pair is given by $O(\sqrt{n})$, there exist $\Omega(\sqrt{n})$
source nodes that can be active simultaneously, and therefore the
total throughput is finally given by (\ref{EQ:achievablerate}),
which completes the proof of the theorem.
\end{proof}

Now it is examined how the upper bound shown in
Section~\ref{SEC:Upperex} is close to the achievable throughput
scaling.

\begin{remark}
Based on Theorems~\ref{TH:Cutset} and~\ref{TH:achievable}, it is
easy to see that the achievable rate and the upper bound are of
exactly the same order. MH is therefore order-optimal in regular
networks with unit node density for all the attenuation regimes.
\end{remark}

We also remark that applying the hierarchical cooperation
strategy~\cite{OzgurLevequeTse:07} may not be helpful to improve the
achievable throughput due to long-range MIMO transmissions, which
severely degrade performance in highly power-limited
networks.\footnote{In wireless radio networks of unit node density,
the hierarchical cooperation provides a near-optimal throughput
scaling for the operating regimes $2<\alpha<3$, where $\alpha$
denotes the path-loss exponent that is greater than
2~\cite{OzgurLevequeTse:07}. Note that the analysis
in~\cite{OzgurLevequeTse:07} is valid under the assumption that
$\alpha$ is kept at the same value on all levels of hierarchy.} To
be specific, at the top level of the hierarchy, the transmissions
between two clusters having distance $O(\sqrt{n})$ become a
bottleneck, and thus cause a significant throughput degradation. It
is further seen that even with the random phase model, which may
enable us to obtain enough DoF gain, the benefit of randomness
cannot be exploited because of the power limitation.

\subsection{Dense Networks}

From the converse result in Section~\ref{SEC:upperdense}, it is seen
that in dense networks, there exists either a bandwidth or power
limitation, or both, according to the path-loss attenuation regimes.
Based on the earlier
studies~\cite{GuptaKumar:00,OzgurLevequeTse:07,OzgurJohariTseLeveque:10,ShinJeonDevroyeVuChungLeeTarokh:08}
for wireless radio networks, it follows that using MH routing is
preferred at power-limited regimes, while the HC strategy may have a
better performance at bandwidth-limited regimes. Thus, existing
schemes need to be used carefully, depending on operating regimes.

In this subsection, the nearest-neighbor MH routing
in~\cite{GuptaKumar:00} is described with a slight modification. The
basic procedure of the MH protocol under our dense regular network
is similar to the extended network case, and is briefly described as
follows:

\begin{itemize}
\item Divide the network into $n$ square routing cells, each of
which has the same area.
\item Draw a line connecting an S--D pair.
\item At each node, use the transmit power of
\begin{equation}
P\min\left\{1,\frac{a(f)^{1/\sqrt{n}}N(f)}{n^{\alpha/2}}\right\}.
\nonumber
\end{equation}
\end{itemize}

The scheme operates with the full power when
$a(f)=\Omega\left(\frac{n^{\alpha\sqrt{n}/2}}{N(f)^{\sqrt{n}}}\right)$.
On the other hand, when
$a(f)=o\left(\frac{n^{\alpha\sqrt{n}/2}}{N(f)^{\sqrt{n}}}\right)$,
the transmit power $Pa(f)^{1/\sqrt{n}}N(f)/n^{\alpha/2}$, which
scales slower than $\Theta(1)$, is sufficient so that the received
SNR at each node is bounded by 1 (note that having a higher power is
unnecessary in terms of throughput improvement).

The amount of interference is now quantified to show the achievable
throughput based on MH.

\begin{lemma} \label{LEM:interferencedense}
Consider a dense regular network that uses the nearest-neighbor MH
protocol. Then, the total interference power $P_I$ from other
simultaneously transmitting nodes, corresponding to the set
$I\subset\{1,\cdots,n\}$, is bounded by
\begin{eqnarray} \label{EQ:PIdense}
P_I= \left\{\begin{array}{lll}
O\left(\frac{\max\left\{n^{(1/2-\beta)(2-\alpha)},\log
n\right\}}{n^{\beta a_5/2}}\right) &\textrm{if $0\le\beta<1/2$}
\\ O\left(n^{-\beta
a_5/2}\right) &\textrm{if
$\beta=1/2$} \\
O\left(\frac{n^{\alpha/2}}{(1+\epsilon_0)^{n^{\beta-1/2}}}\right)
&\textrm{if $\beta>1/2$}
\end{array}\right.
\end{eqnarray}
for an arbitrarily small $\epsilon_0>0$, where $a_5$ and $\beta$ are
defined in (\ref{EQ:Nfeq}) and (\ref{EQ:afepsilon}), respectively.
\end{lemma}

The proof of this lemma is presented in
Appendix~\ref{PF:interferencedense}. From (\ref{EQ:Nfaf}) and
(\ref{EQ:afepsilon}), we note that when $\beta=1/2$, it follows that
$P_I=O(N(f))$, i.e., $P_I$ is upper-bounded by the PSD $N(f)$ of
noise. Using Lemma~\ref{LEM:interferencedense}, a lower bound on the
capacity scaling can be derived, and hence the following result
shows the achievable rates under the MH protocol in a dense network.

\begin{theorem} \label{THM:Tndenselower}
In an underwater regular network of unit area,
\begin{eqnarray} \label{EQ:Tndenselower}
T(n)=\left\{\begin{array}{lll}
\Omega\left(\frac{\sqrt{n}}{\max\left\{n^{(1/2-\beta)(2-\alpha)},\log
n\right\}}\right) &\textrm{if $0\le\beta<1/2$}
\\ \Omega\left(\sqrt{n}\right) &\textrm{if
$\beta=1/2$} \\
\Omega\left(\frac{n^{(1+\alpha+\beta
a_5)/2}}{(1+\epsilon_0)^{n^{\beta-1/2}}}\right) &\textrm{if
$\beta>1/2$}
\end{array}\right.
\end{eqnarray}
is achievable.
\end{theorem}

\begin{proof}
Suppose that the nearest-neighbor MH protocol described above is
used. Then, from (\ref{EQ:af}), the received signal power $P_r$ from
the desired transmitter is given by
\begin{eqnarray}
P_r =
\frac{P\min\left\{1,\frac{a(f)^{1/\sqrt{n}}N(f)}{n^{\alpha/2}}\right\}n^{\alpha/2}}{c_0a(f)^{1/\sqrt{n}}},
\nonumber
\end{eqnarray}
which can be rewritten as
\begin{eqnarray} \label{EQ:Prdense}
&&\frac{Pn^{\alpha/2}}{c_0(1+\epsilon_0)^{n^{\beta-1/2}}}\min\left\{1,\frac{(1+\epsilon_0)^{n^{\beta-1/2}}}{n^{(\alpha+\beta
a_5)/2}}\right\} \nonumber\\ =\!\!\!\!\!\!\!&&
\left\{\begin{array}{lll} \frac{P}{c_0n^{\beta a_5/2}} &\textrm{if
$0\le\beta\le1/2$}
\\ \frac{Pn^{\alpha/2}}{c_0(1+\epsilon_0)^{n^{\beta-1/2}}} &\textrm{if
$\beta>1/2$}
\end{array}\right.
\end{eqnarray}
with respect to the parameter $\beta$ using (\ref{EQ:Nfaf}) and
(\ref{EQ:afepsilon}). A lower bound on the throughput is now
obtained using (\ref{EQ:PIdense}) and (\ref{EQ:Prdense}). By
assuming the worst case noise, which lower-bounds the transmission
rate, and full CSI at the receiver, the achievable throughput per
S--D pair is then lower-bounded by
\begin{eqnarray}
&&\log(1+\text{SINR}) \nonumber\\ =\!\!\!\!\!\!\!&&
\log\left(1+\frac{P_r}{N(f)+P_I}\right) \nonumber\\
=\!\!\!\!\!\!\!&& \left\{\begin{array}{lll}
\Omega\left(\log\left(1+\frac{1}{\max\left\{n^{(1/2-\beta)(2-\alpha)},\log
n\right\}}\right)\right) &\textrm{if $0\le\beta<1/2$}
\\ \Omega(1) &\textrm{if
$\beta=1/2$} \\
\Omega\left(\log\left(1+\frac{n^{(\alpha+\beta
a_5)/2}}{(1+\epsilon_0)^{n^{\beta-1/2}}}\right)\right) &\textrm{if
$\beta>1/2$}
\end{array}\right. \nonumber\\
=\!\!\!\!\!\!\!&& \left\{\begin{array}{lll}
\Omega\left(\frac{1}{\max\left\{n^{(1/2-\beta)(2-\alpha)},\log
n\right\}}\right) &\textrm{if $0\le\beta<1/2$}
\\ \Omega(1) &\textrm{if
$\beta=1/2$} \\
\Omega\left(\frac{n^{(\alpha+\beta
a_5)/2}}{(1+\epsilon_0)^{n^{\beta-1/2}}}\right) &\textrm{if
$\beta>1/2$,}
\end{array}\right. \nonumber
\end{eqnarray}
where the second equality holds since $N(f)=\Theta(n^{-\beta
a_5/2})$ and thus $P_I=\omega(N(f))$ for $0\le\beta<1/2$,
$P_I=\Theta(P_r)=\Theta(N(f))$ for $\beta=1/2$, and $P_I=o(N(f))$
for $\beta>1/2$. The last equality comes from the fact that
$\log(1+x)=(\log e)x+O(x^2)$ for small $x>0$. Since there are
$\Omega(\sqrt{n})$ S--D pairs that can be active simultaneously in
the network, the total throughput is finally given by
(\ref{EQ:Tndenselower}), which completes the proof.
\end{proof}

Note that the achievable throughput~\cite{GuptaKumar:00} for
wireless radio networks of unit area using MH routing is given by
$\Omega(\sqrt{n})$, which is the same as the case for which
$\beta=1/2$ (or equivalently
$a(f)=\Theta\left((1+\epsilon_0)^{\sqrt{n}}\right)$) in a dense
underwater network. The lower bound on the total throughput $T(n)$
is also shown in Fig.~\ref{FIG:bounds} according to the parameter
$\beta$. From this result, an interesting observation follows. To be
specific, in the regime $0\le\beta\le1/2$, the lower bound on $T(n)$
grows linearly with increasing $\beta$, because the total
interference power $P_I$ in (\ref{EQ:PIdense}) tends to decrease as
$\beta$ increases. In this regime, note that $P_I=\Omega(P_r)$.
Meanwhile, when $\beta>1/2$, the lower bound reduces rapidly due to
the exponential path-loss attenuation in terms of increasing
$\beta$.

In addition, similarly as in Section~\ref{SEC:upperdense}, two other
expressions on the achievability result are now summarized as in the
following.

\begin{remark}
From (\ref{EQ:Nfaf}) and (\ref{EQ:afepsilon}), the lower bound on
the throughput $T(n)$ and the corresponding operating regimes can
also be presented below in terms of the attenuation parameter
$a(f)$:
\begin{eqnarray}
T(n)=\left\{\begin{array}{lll}
\Omega\left(\frac{\sqrt{n}}{\max\left\{a(f)^{(2-\alpha)/\sqrt{n}},\log
n\right\}}\right) &\textrm{if $a(f)=\Omega(1)$ and
$a(f)=o\left((1+\epsilon_0)^{\sqrt{n}}\right)$}
\\ \Omega\left(\sqrt{n}\right) &\textrm{if
$a(f)=\Theta\left((1+\epsilon_0)^{\sqrt{n}}\right)$} \\
\Omega\left(\frac{n^{(1+\alpha)/2}(\log
a(f))^{a_5/2}}{a(f)^{1/\sqrt{n}}}\right) &\textrm{if
$a(f)=\omega\left((1+\epsilon_0)^{\sqrt{n}}\right)$.}
\end{array}\right. \nonumber
\end{eqnarray}
Furthermore, using (\ref{EQ:afTheta}) and (\ref{EQ:NfTheta}) follows
that
\begin{eqnarray}
T(n)=\left\{\begin{array}{lll}
\Omega\left(\frac{\sqrt{n}}{\max\left\{e^{c_1(2-\alpha)f^2/\sqrt{n}},\log
n\right\}}\right) &\textrm{if $f=\Omega(1)$ and
$f=o\left(n^{1/4}\right)$}
\\ \Omega\left(\sqrt{n}\right) &\textrm{if
$f=\Theta\left(n^{1/4}\right)$} \\
\Omega\left(\frac{n^{(1+\alpha)/2}f^{a_5}}{e^{c_1f^2/\sqrt{n}}}\right)
&\textrm{if $f=\omega\left(n^{1/4}\right)$,}
\end{array}\right. \nonumber
\end{eqnarray}
which represents the lower bound for the operating regimes obtained
from the relationship between the frequency $f$ and the number $n$
of nodes.
\end{remark}

Now let us turn to examining how the upper bound shown in
Section~\ref{SEC:upperdense} is close to the achievable throughput
scaling.

\begin{remark}
Based on Theorems~\ref{THM:upperdense} and~\ref{THM:Tndenselower},
it is seen that if $\beta\ge1/2$, then the achievable rate of the MH
protocol is close to the upper bound up to $n^{\epsilon}$ for an
arbitrarily small $\epsilon>0$ (note that the two bounds are of
exactly the same order especially for $\beta>1/2$). The condition
$\beta\ge1/2$ corresponds to the high path-loss attenuation regime,
and is equivalent to
$a(f)=\Omega\left((1+\epsilon_0)^{\sqrt{n}}\right)$ or
$f=\Omega\left(n^{1/4}\right)$. Therefore, the MH is order-optimal
in regular networks of unit area under the aforementioned operating
regimes, whereas in extended networks, using MH routing results in
the order optimality for all the operating regimes.
\end{remark}

Finally, we remark that applying the HC
strategy~\cite{OzgurLevequeTse:07} does not guarantee the order
optimality in the regime $0\le\beta<1/2$ (i.e., low and medium
path-loss attenuation regimes). The primary reason is specified
under each operating regime: for the condition $\beta=0$, following
the steps similar to those of Lemma~\ref{LEM:interferencedense}, it
follows that $P_I=\omega(P_r)$ at all levels of the hierarchy,
thereby resulting in $\text{SINR}=o(1)$ for each transmission (the
details are not shown in this paper). It is thus not possible to
achieve a linear throughput scaling. Now let us focus on the case
where $0<\beta<1/2$. At the top level of the hierarchy, the
transmissions between two clusters having distance $O(1)$ becomes a
bottleneck because of the exponential path-loss attenuation with
propagation distance. Hence, the achievable throughput of the HC
decays exponentially with respect to $n$, which is significantly
lower than that in (\ref{EQ:Tndenselower}).


\section{Extension to Random Networks} \label{SEC:Random}

In this section, we would like to mention a random network
configuration, where $n$ S--D pairs are uniformly and independently
distributed on a square of unit node density (i.e., an extended
random network).

We first discuss an upper bound for extended random networks. A
precise upper bound can be obtained using the binning argument
of~\cite{OzgurLevequeTse:07} (refer to Appendix V
in~\cite{OzgurLevequeTse:07} for the details). Consider the same cut
$L$, which divides the network area into two halves, as in the
regular network case. For analytical convenience, we can
artificially assume the empty zone $E_L$, in which there are no
nodes in the network, consisting of a rectangular slab of width
$0<\bar{c}<\frac{1}{\sqrt{7}e^{1/4}}$, independent of $n$,
immediately to the right of the centerline (cut), as done
in~\cite{OzgurJohariTseLeveque:10} (see
Fig.~\ref{FIG:displacement2}).\footnote{Although this assumption
does not hold in our random configuration, it is shown
in~\cite{OzgurJohariTseLeveque:10} that there exists a vertical cut
such that there are no nodes located closer than
$0<\bar{c}<\frac{1}{\sqrt{7}e^{1/4}}$ on both sides of this cut when
we allow a cut that is not necessarily linear. Such an existence is
proved by using percolation
theory~\cite{MeesterRoy:96,FranceschettiDouseTseThiran:07}. This
result can be directly applied to our network model since it only
relies on the node distribution but not the channel characteristics.
Hence, removing the assumption does not cause any change in
performance.} Let us state the following lemma.

\begin{lemma} \label{LEM:nodenum}
Assume a two dimensional extended network where $n$ nodes are
uniformly distributed. When the network area is divided into $n$
squares of unit area, there are fewer than $\log n$ nodes in each
square with high probability.
\end{lemma}

Since the result in Lemma~\ref{LEM:nodenum} depends on the node
distribution but not the channel characteristics, the proof
essentially follows that presented
in~\cite{FranceschettiDouseTseThiran:07}. By
Lemma~\ref{LEM:nodenum}, we now take into account the network
transformation resulting in a regular network with at most $\log n$
and $2\log n$ nodes, on the left and right, respectively, at each
square vertex except for the empty zone (see
Fig.~\ref{FIG:displacement2}). Then, the nodes in each square are
moved together onto one vertex of the corresponding square. More
specisely, under the cut $L$, the node displacement is performed in
the sense of decreasing the Euclidean distance between source node
$i\in S_L$ and the corresponding destination $k\in D_L$, as shown in
Fig.~\ref{FIG:displacement2}, which will provide an upper bound on
$P_L^{(k)}$ in (\ref{EQ:PL}). It is obviously seen that the amount
of power transfer under the transformed regular network is greater
than that under another regular network with at most $\log n$ nodes
at each vertex, located at integer lattice positions in a square
region of area $n$. Hence, the upper bound for random networks is
boosted by at least a logarithmic factor of $n$ compared to that of
regular networks discussed in Section~\ref{SEC:Upper}.

Now we turn our attention to showing an achievable throughput for
extended random networks. In this case, the nearest-neighbor MH
protocol~\cite{GuptaKumar:00} can also be utilized since our network
is highly power-limited. Then, the area of each routing cell needs
to scale with $2\log n$ to guarantee at least one node in a
cell~\cite{GuptaKumar:00,ElGamalMammenPrabhakarShah:06}.\footnote{When
methods from percolation theory are applied to our random
network~\cite{MeesterRoy:96,FranceschettiDouseTseThiran:07}, the
routing area constructed during the highway phase is a certain
positive constant that is less than 1 and independent of $n$. The
distance in the draining and delivery phases, corresponding to the
first and last hops of a packet transmission, respectively, is
nevertheless given by some constant times $\log n$, thereby limiting
performance, especially for the condition $a(f)=\omega(1)$. Hence,
using the routing protocol in~\cite{FranceschettiDouseTseThiran:07}
indeed does not perform better than the conventional MH
case~\cite{GuptaKumar:00} in random networks.} Each routing cell
operates based on 9-time division multiple access to avoid causing
large interference to its neighboring
cells~\cite{GuptaKumar:00,ElGamalMammenPrabhakarShah:06}. For the MH
routing, since per-hop distance is given by $\Theta(\sqrt{\log n})$,
the received signal power from the intended transmitter and the SINR
seen by any receiver are expressed as
\begin{equation}
\frac{c_{10}P}{(\log n)^{\alpha/2}a(f)^{\delta\sqrt{\log n}}}
\nonumber
\end{equation}
and
\begin{equation}
\Omega\left(\frac{1}{(\log n)^{\alpha/2}a(f)^{\delta\sqrt{\log
n}}N(f)}\right), \nonumber
\end{equation}
respectively, for some constants $c_{10}>0$ and $\delta\ge\sqrt{2}$
independent of $n$. Since the number of hops per S--D pair is given
by $O(\sqrt{n/\log n})$, there exist $\Omega(\sqrt{n/\log n})$
simultaneously active sources, and thus the total achievable
throughput $T(n)$ is finally given by
\begin{equation}
T(n)=\Omega\left(\frac{n^{1/2}}{(\log
n)^{(\alpha+1)/2}a(f)^{\delta\sqrt{\log n}}N(f)}\right) \nonumber
\end{equation}
for some constant $\delta\ge\sqrt{2}$ independent of $n$ (note that
this relies on the fact that $\log(1+x)$ can be approximated by $x$
for small $x>0$). Hence, using the MH protocol results in at least a
polynomial decrease in the throughput compared to the regular
network case shown in Section~\ref{SEC:Achievability}.\footnote{In
terrestrial radio channels, there is a logarithmic gap in the
achievable scaling laws between regular and random
networks~\cite{GuptaKumar:00,XieKumar:04}.} This comes from the fact
that the received signal power tends to be mainly limited due to
exponential attenuation with transmission distance
$\Theta(\sqrt{\log n})$. Note that in underwater networks,
randomness on the node distribution causes a huge performance
degradation on the throughput scaling. Therefore, we may conclude
that the existing MH scheme does not satisfy the order optimality
under extended random networks regardless of the attenuation
parameter $a(f)$.


\section{Conclusion} \label{SEC:Conc}

The attenuation parameter and the capacity scaling laws have been
characterized in a narrow-band channel of underwater acoustic
networks. Provided that the frequency $f$ scales relative to the
number $n$ of nodes, the information-theoretic upper bounds and the
achievable throughputs were obtained as functions of the attenuation
parameter $a(f)$ in regular networks. In extended networks, based on
the power transfer argument, the upper bound was shown to decrease
in inverse proportion to $a(f)$. In dense networks, the upper bound
was derived characterizing three different operating regimes, in
which there exists either a bandwidth or power limitation, or both.
In addition, to show the achievability result, the nearest-neighbor
MH protocol was introduced with a simple modification, and its
throughput scaling was analyzed. We proved that the MH protocol is
order-optimal in all operating regimes of extended networks and in
power-limited regimes (i.e., the case where the frequency $f$ scales
faster than or as $n^{1/4}$) of dense networks. Therefore, it turned
out that there exists a right frequency scaling that makes our
scaling results for underwater acoustic networks to break free from
scaling limitations related to the channel characteristics that were
described in~\cite{LucaniMedardStojanovic:08}. Our scaling results
were also extended to the random network scenario, where it was
shown that the conventional MH scheme does not satisfy the order
optimality for all the operating regimes.


\appendix

\section{Appendix}

\subsection{Proof of Lemma~\ref{LEM:proper}} \label{PF:proper}

The following definition is used to simply provide the proof.

{\em Definition 1~\cite{NeeserMassey:93}:} A complex random variable
$Y$ is said to be proper if $\tilde{\Sigma}_Y=0$, where
$\tilde{\Sigma}_Y$, called the pseudo-covariance, is given by
$E[(Y-E[Y])^2]$.

Since the $(k,i)$-th element of the channel matrix ${\bf H}_L$ is
given by (\ref{EQ:hki}), it follows that
\begin{eqnarray}
E\left[\left(h_{ki}-E[h_{ki}]\right)^2\right]=\frac{1}{A(r_{ki},f)}E\left[\left(e^{j\theta_{ki}}-E\left[e^{j\theta_{ki}}\right]\right)^2\right].
\nonumber
\end{eqnarray}
From the fact that
\begin{eqnarray}
E\left[e^{j\theta_{ki}}\right]=E\left[\cos(\theta_{ki})+j\sin(\theta_{ki})\right]=0
\nonumber
\end{eqnarray}
due to uniformly distributed $\theta_{ki}$ over $[0, 2\pi]$, we thus
have
\begin{eqnarray}
E\left[\left(h_{ki}-E[h_{ki}]\right)^2\right]\!\!\!\!\!\!\!&&=\frac{1}{A(r_{ki},f)}E\left[e^{j2\theta_{ki}}\right]
\nonumber\\
&&=\frac{1}{A(r_{ki},f)}E\left[\cos(2\theta_{ki})+j\sin(2\theta_{ki})\right] \nonumber\\
&&=0, \label{EQ:hkiproper}
\end{eqnarray}
which complete the proof, because (\ref{EQ:hkiproper}) holds for all
$i\in S_L$ and $k\in D_L$.


\subsection{Proof of Lemma~\ref{LEM:PL}} \label{PF:PL}

An upper bound on $P_L^{(k)}$ can be found by using the
node-indexing and layering techniques similar to those shown in
Section VI of~\cite{ShinJeonDevroyeVuChungLeeTarokh:08}. As
illustrated in Fig.~\ref{FIG:displacement}, layers are introduced,
where the $i$-th layer $l_i'$ of the network represents the ring
with width 1 drawn based on a destination node $k\in D_L$, whose
coordinate is given by $(k_x,k_y)$, where
$i\in\{1,\cdots,\sqrt{n}\}$. More precisely, the ring is enclosed by
the circumferences of two circles, each of which has radius $k_x+i$
and $k_x+i-1$, respectively, at its same center (see
Fig.~\ref{FIG:displacement}). Then from (\ref{EQ:PL}), the term
$P_{L}^{(k)}$ is given by
\begin{equation}
P_{L}^{(k)}=
\frac{P}{c_0}\sum_{i_x=1}^{\sqrt{n}/2}\sum_{i_y=1}^{\sqrt{n}}\frac{1}{\left((i_x+k_x-1)^2+(i_y-k_y)^2\right)^{\alpha/2}a(f)^{\sqrt{(i_x+k_x-1)^2+(i_y-k_y)^2}}}.
\nonumber
\end{equation}
It is further assumed that all the nodes in each layer are moved
onto the innermost boundary of the corresponding ring, which
provides an upper bound for $P_{L}^{(k)}$. From the fact that there
exist $\Theta(k_x+i)$ nodes in the layer $l_i'$ since the area of
$l_i'$ is given by $\pi(2k_x+2i-1)$, $P_{L}^{(k)}$ is then
upper-bounded by
\begin{eqnarray} \label{EQ:dLupp}
P_{L}^{(k)}\!\!\!\!\!\!\!
&&\le\frac{P}{c_0}\sum_{i'=k_x}^{\infty}\frac{c_{11}(i'+1)}{i'^{\alpha}a(f)^{i'}}
\nonumber\\
&&\le\frac{2c_{11}P}{c_0k_x^{\alpha-1}}\sum_{i'=k_x}^{\infty}\frac{1}{a(f)^{i'}}
\nonumber\\ &&\le
\frac{2c_{11}P}{c_0k_x^{\alpha-1}}\left(\frac{1}{a(f)^{k_x}}+\int_{k_x}^{\infty}\frac{1}{a(f)^x}dx\right) \nonumber\\
&& \le\frac{c_{12}P}{k_x^{\alpha-1}a(f)^{k_x}} \nonumber
\end{eqnarray}
for some positive constants $c_0$, $c_{11}$, and $c_{12}$
independent of $n$, where the fourth inequality holds since
$a(f)>1$, which finally yields (\ref{EQ:PLupper}). This completes
the proof.


\subsection{Proof of Lemma~\ref{LEM:PLdense}} \label{PF:PLdense}

Upper and lower bounds on $P_L^{(k)}$ in a dense network are derived
by basically following the same node indexing and layering
techniques as those in Appendix~\ref{PF:PL}. We refer to
Fig.~\ref{FIG:displacement} for the detailed description (note that
the destination nodes are, however, located at positions
$\left(\frac{k_x}{\sqrt{n}}, \frac{k_y}{\sqrt{n}}\right)$ in dense
networks). Similarly to the extended network case, from
(\ref{EQ:PL}), the term $P_{L}^{(k)}$ is then given by
\begin{equation}
P_{L}^{(k)}=
\frac{P}{c_0}\sum_{i_x=1}^{\sqrt{n}/2}\sum_{i_y=1}^{\sqrt{n}}\frac{n^{\alpha/2}}{\left((i_x+k_x-1)^2+(i_y-k_y)^2\right)^{\alpha/2}a(f)^{\sqrt{\left((i_x+k_x-1)^2+(i_y-k_y)^2\right)/n}}}.
\nonumber
\end{equation}

First, focus on how to obtain an upper bound for $P_{L}^{(k)}$.
Assuming that all the nodes in each layer are moved onto the
innermost boundary of the corresponding ring, we then have
\begin{eqnarray}
P_{L}^{(k)} \!\!\!\!\!\!\!&&\le
\frac{Pn^{\alpha/2}}{c_0}\sum_{i'=k_x}^{k_x+\sqrt{n}/2-1}\frac{c_{11}(i'+1)}{i'^{\alpha}a(f)^{i'\sqrt{n}}}
\nonumber\\ && \le
\frac{2c_{11}Pn^{\alpha/2}}{c_0}\sum_{i'=k_x}^{\sqrt{n}}\frac{1}{i'^{\alpha-1}a(f)^{i'\sqrt{n}}}
\nonumber\\ && =
c_{12}Pn^{\alpha/2}\sum_{i'=k_x}^{\sqrt{n}}\frac{1}{i'^{\alpha-1}\left(1+\epsilon_0\right)^{i'n^{\beta-1/2}}}
\label{EQ:PLupperapp}
\end{eqnarray}
for some positive constants $c_0$, $c_{11}$, and $c_{12}$
independent of $n$ and an arbitrarily small $\epsilon_0>0$, where
the equality comes from the relationship (\ref{EQ:afepsilon})
between $a(f)$ and $\beta$. We first consider the case where
$k_x=o\left(n^{1/2-\beta+\epsilon}\right)$ for an arbitrarily small
$\epsilon>0$.
Under this condition, from the fact that the term $i'^{\alpha-1}$ in
the RHS of (\ref{EQ:PLupperapp}) is dominant in terms of
upper-bounding $P_{L}^{(k)}$ for $i'=k_x,\cdots,\sqrt{n}$,
(\ref{EQ:PLupperapp}) is further bounded by
\begin{eqnarray}
P_{L}^{(k)} \!\!\!\!\!\!\!&&\le
c_{12}Pn^{\alpha/2}\sum_{i'=k_x}^{\sqrt{n}}\frac{1}{i'^{\alpha-1}}
\nonumber\\ && \le
c_{12}Pn^{\alpha/2}\left(\frac{1}{k_x^{\alpha-1}}+\int_{k_x}^{\sqrt{n}}\frac{1}{x^{\alpha-1}}dx\right),
\nonumber
\end{eqnarray}
which yields
$P_{L}^{(k)}=O\left(n^{\alpha/2}(\sqrt{n})^{2-\alpha}\right)=O(n)$
for $1\le\alpha<2$ and $P_{L}^{(k)}=O\left(n\log n\right)$ for
$\alpha=2$. When $k_x=\Omega(n^{1/2-\beta+\epsilon})$, the upper
bound (\ref{EQ:PLupperapp}) for $P_{L}^{(k)}$ is dominated by the
term $\left(1+\epsilon_0\right)^{i'n^{\beta-1/2}}$, and thus is
given by
\begin{eqnarray}
P_{L}^{(k)} \!\!\!\!\!\!\!&&\le
c_{12}Pn^{\alpha/2}\sum_{i'=k_x}^{\sqrt{n}}\frac{1}{\left(1+\epsilon_0\right)^{i'n^{\beta-1/2}}}
\nonumber\\ && \le
c_{12}Pn^{\alpha/2}\left(\frac{1}{\left(1+\epsilon_0\right)^{k_x
n^{\beta-1/2}}}+\int_{k_x}^{\sqrt{n}}\frac{1}{\left(1+\epsilon_0\right)^{x
n^{\beta-1/2}}}dx\right) \nonumber\\ && =
c_{12}Pn^{\alpha/2}\left(\frac{1}{\left(1+\epsilon_0\right)^{k_x
n^{\beta-1/2}}}+\int_{1}^{\sqrt{n}/k_x}\frac{k_x}{\left(1+\epsilon_0\right)^{x
k_x n^{\beta-1/2}}}dx\right) \nonumber\\ &&\le
c_{12}Pn^{\alpha/2}\left(\frac{1}{\left(1+\epsilon_0\right)^{k_x
n^{\beta-1/2}}}+\frac{n^{1/2-\beta}}{\left(1+\epsilon_0\right)^{k_x
n^{\beta-1/2}}}\right) \nonumber\\ &&\le
\frac{c_{13}Pn^{\alpha/2}}{\left(1+\epsilon_0\right)^{k_x
n^{\beta-1/2}}}\max\left\{1,n^{1/2-\beta}\right\}
\label{EQ:PLlowerthird}
\end{eqnarray}
for some constant $c_{13}>0$ independent of $n$, which is the last
result in (\ref{EQ:PLupperdense}).

Next, let us turn to deriving a lower bound for $P_{L}^{(k)}$. Since
each layer has at least one node that is onto the innermost boundary
of the corresponding ring, the lower bound similarly follows
\begin{eqnarray}
P_{L}^{(k)} \!\!\!\!\!\!\!&&\ge
\frac{Pn^{\alpha/2}}{c_0}\sum_{i'=k_x}^{k_x+\sqrt{n}/2-1}\frac{1}{i'^{\alpha}a(f)^{i'\sqrt{n}}}
\nonumber\\ && =
c_{12}Pn^{\alpha/2}\sum_{i'=k_x}^{k_x+\sqrt{n}/2-1}\frac{1}{i'^{\alpha}\left(1+\epsilon_0\right)^{i'n^{\beta-1/2}}}.
\label{EQ:PLlowerapp}
\end{eqnarray}
For the condition $k_x=o(n^{1/2-\beta+\epsilon})$,
(\ref{EQ:PLlowerapp}) is represented as
\begin{eqnarray}
P_{L}^{(k)} \!\!\!\!\!\!\!&&\ge
c_{12}Pn^{\alpha/2}\sum_{i'=k_x}^{k_x+\sqrt{n}/2-1}\frac{1}{i'^{\alpha}\left(1+\epsilon_0\right)^{i'n^{\beta-1/2}}}
\nonumber\\ &&\ge
c_{12}Pn^{\alpha/2}\sum_{i'=k_x}^{2k_x-1}\frac{1}{i'^{\alpha}\left(1+\epsilon_0\right)^{i'n^{\beta-1/2}}}
\nonumber\\ &&\ge
\frac{c_{12}Pn^{\alpha/2}}{n^{\epsilon'}}\sum_{i'=k_x}^{2k_x-1}\frac{1}{i'^{\alpha}}
\nonumber\\ && \ge
\frac{c_{14}Pn^{\alpha/2}}{n^{\epsilon'}k_x^{\alpha-1}} \nonumber
\end{eqnarray}
for an arbitrarily small $\epsilon'>0$ and some constant $c_{14}>0$
independent of $n$, where the third inequality holds due to
$\left(1+\epsilon_0\right)^{k_x n^{\beta-1/2}}=O(n^{\epsilon'})$. On
the other hand, similarly as in the steps of
(\ref{EQ:PLlowerthird}), the condition
$k_x=\Omega(n^{1/2-\beta+\epsilon})$ yields the following lower
bound for $P_{L}^{(k)}$:
\begin{eqnarray}
P_{L}^{(k)} \!\!\!\!\!\!\!&&\ge c_{12}P
\sum_{i'=k_x}^{k_x+\sqrt{n}/2-1}\frac{1}{\left(1+\epsilon_0\right)^{i'n^{\beta-1/2}}}
\nonumber\\ && \ge c_{12}P
\left(\frac{1}{\left(1+\epsilon_0\right)^{k_x
n^{\beta-1/2}}}+\int_{k_x+1}^{k_x+\sqrt{n}/2-1}\frac{1}{\left(1+\epsilon_0\right)^{x
n^{\beta-1/2}}}dx\right) \nonumber\\ &&\ge
c_{15}P\left(\frac{1}{\left(1+\epsilon_0\right)^{k_x
n^{\beta-1/2}}}+\frac{n^{1/2-\beta}}{\left(1+\epsilon_0\right)^{(k_x+1)
n^{\beta-1/2}}}\right) \nonumber\\ &&\ge
\frac{c_{15}P}{\left(1+\epsilon_0\right)^{k_x
n^{\beta-1/2}}}\max\left\{1,\frac{n^{1/2-\beta}}{\left(1+\epsilon_0\right)^{n^{\beta-1/2}}}\right\}
\nonumber
\end{eqnarray}
some constant $c_{15}>0$ independent of $n$, which finally complete
the proof of the lemma.


\subsection{Proof of Lemma~\ref{LEM:interferencedense}} \label{PF:interferencedense}

The layering technique illustrated in Fig.~\ref{FIG:layer} is
applied as in the extended network case. From (\ref{EQ:af}), the
total interference power $P_I$ at each node from simultaneously
transmitting nodes is then upper-bounded by
\begin{eqnarray}
P_I\!\!\!\!\!\!\!&&=\sum_{k=1}^{\sqrt{n}}(8k)\frac{P\min\left\{1,\frac{a(f)^{1/\sqrt{n}}N(f)}{n^{\alpha/2}}\right\}}{c_0(k/\sqrt{n})^\alpha
a(f)^{k/\sqrt{n}}} \nonumber\\ &&=
\frac{8Pn^{\alpha/2}\min\left\{1,\frac{a(f)^{1/\sqrt{n}}N(f)}{n^{\alpha/2}}\right\}}{c_0}\sum_{k=1}^{\sqrt{n}}\frac{1}{k^{\alpha-1}a(f)^{k/\sqrt{n}}}.
\label{EQ:interferencedense}
\end{eqnarray}
Using (\ref{EQ:Nfaf}) and (\ref{EQ:afepsilon}), the upper bound
(\ref{EQ:interferencedense}) on $P_I$ can be expressed as
\begin{eqnarray}
P_I\le\!\!\!\!\!\!\!&&
c_{16}Pn^{\alpha/2}\min\left\{1,\frac{(1+\epsilon_0)^{n^{\beta-1/2}}}{n^{(\alpha+\beta
a_5)/2}}\right\}\sum_{k=1}^{\sqrt{n}}\frac{1}{k^{\alpha-1}(1+\epsilon_0)^{kn^{\beta-1/2}}}
\nonumber\\ \le\!\!\!\!\!\!\!&& \left\{\begin{array}{lll}
\frac{c_{16}P}{n^{\beta
a_5/2}}\sum_{k=1}^{\sqrt{n}}\frac{1}{k^{\alpha-1}(1+\epsilon_0)^{kn^{\beta-1/2}}}
&\textrm{if $0\le\beta<1/2$}
\\ \frac{c_{16}P}{n^{\beta
a_5/2}}\sum_{k=1}^{\sqrt{n}}\frac{1}{(1+\epsilon_0)^{k}} &\textrm{if
$\beta=1/2$} \\
c_{16}Pn^{\alpha/2}\sum_{k=1}^{\sqrt{n}}\frac{1}{(1+\epsilon_0)^{kn^{\beta-1/2}}}
&\textrm{if $\beta>1/2$}
\end{array}\right. \nonumber\\ \le\!\!\!\!\!\!\!&& \left\{\begin{array}{lll}
\frac{c_{16}P}{n^{\beta
a_5/2}}\sum_{k=1}^{\sqrt{n}}\frac{1}{k^{\alpha-1}(1+\epsilon_0)^{kn^{\beta-1/2}}}
&\textrm{if $0\le\beta<1/2$}
\\ \frac{c_{16}P}{n^{\beta
a_5/2}}\frac{1}{(1+\epsilon_0)-1} &\textrm{if
$\beta=1/2$} \\
c_{16}Pn^{\alpha/2}\frac{1}{(1+\epsilon_0)^{n^{\beta-1/2}}-1}
&\textrm{if $\beta>1/2$}
\end{array}\right. \nonumber\\ \le\!\!\!\!\!\!\!&& \left\{\begin{array}{lll}
\frac{c_{16}P}{n^{\beta
a_5/2}}\sum_{k=1}^{\sqrt{n}}\frac{1}{k^{\alpha-1}(1+\epsilon_0)^{kn^{\beta-1/2}}}
&\textrm{if $0\le\beta<1/2$}
\\ \frac{c_{17}P}{n^{\beta
a_5/2}} &\textrm{if
$\beta=1/2$} \\
\frac{c_{17}Pn^{\alpha/2}}{(1+\epsilon_0)^{n^{\beta-1/2}}}
&\textrm{if $\beta>1/2$} \nonumber\\
\end{array}\right.
\end{eqnarray}
for some positive constants $c_{16}$ and $c_{17}$ independent of
$n$. Based on the argument in Appendix~\ref{PF:PLdense}, when
$0\le\beta<1/2$, it follows that
\begin{eqnarray}
\sum_{k=1}^{\sqrt{n}}\frac{1}{k^{\alpha-1}(1+\epsilon_0)^{kn^{\beta-1/2}}}
\!\!\!\!\!\!\!&&=
\sum_{k=1}^{n^{1/2-\beta}-1}\frac{1}{k^{\alpha-1}(1+\epsilon_0)^{kn^{\beta-1/2}}}
+
\sum_{k=n^{1/2-\beta}}^{\sqrt{n}}\frac{1}{k^{\alpha-1}(1+\epsilon_0)^{kn^{\beta-1/2}}}
\nonumber\\ && \le
\sum_{k=1}^{n^{1/2-\beta}-1}\frac{1}{k^{\alpha-1}} +
\frac{1}{n^{(1/2-\beta)(\alpha-1)}}\sum_{k=n^{1/2-\beta}}^{\sqrt{n}}\frac{1}{(1+\epsilon_0)^{kn^{\beta-1/2}}}
\nonumber\\ && \le
\left(1+\int_1^{n^{1/2-\beta}}\frac{1}{x^{\alpha-1}}dx\right)
\nonumber\\ && +
\frac{1}{n^{(1/2-\beta)(\alpha-1)}}\left(\frac{1}{1+\epsilon_0}+\int_{n^{1/2-\beta}}^{\sqrt{n}}\frac{1}{(1+\epsilon_0)^{xn^{\beta-1/2}}}dx\right)
\nonumber\\ && \le 2\int_1^{n^{1/2-\beta}}\frac{1}{x^{\alpha-1}}dx +
\frac{2}{n^{(1/2-\beta)(\alpha-1)}}\int_{1}^{n^\beta}\frac{n^{1/2-\beta}}{(1+\epsilon_0)^{x}}dx
\nonumber\\ && \le \left\{\begin{array}{lll}
4n^{(1/2-\beta)(2-\alpha)} &\textrm{if $1\le\alpha<2$}
\\ \log n &\textrm{if
$\alpha=2$,}
\end{array}\right. \nonumber
\end{eqnarray}
which results in
$P_I=O\left(\frac{\max\left\{n^{(1/2-\beta)(2-\alpha)},\log
n\right\}}{n^{\beta a_5/2}}\right)$. This completes the proof of the
lemma.


\newpage


\begin{figure}[t!]
  \begin{center}
  \leavevmode \epsfxsize=0.38\textwidth   
  \leavevmode 
  \epsffile{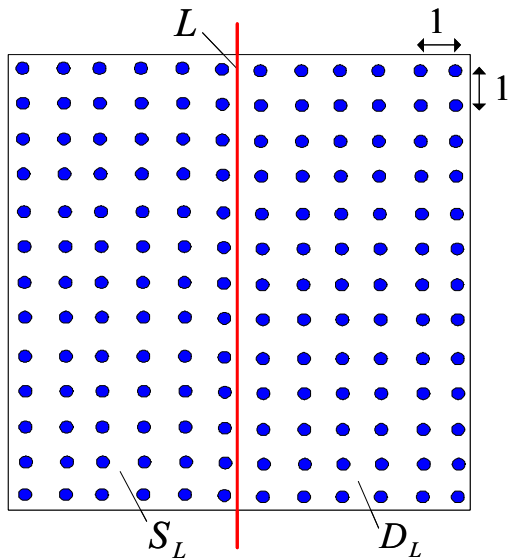}
  \caption{The cut $L$ in a two-dimensional extended regular network. $S_L$ and $D_L$ represent the sets of source and destination nodes, respectively.}
  \label{FIG:cut1}
  \end{center}
\end{figure}

\begin{figure}[t!]
  \begin{center}
  \leavevmode \epsfxsize=0.41 \textwidth   
  \leavevmode 
  \epsffile{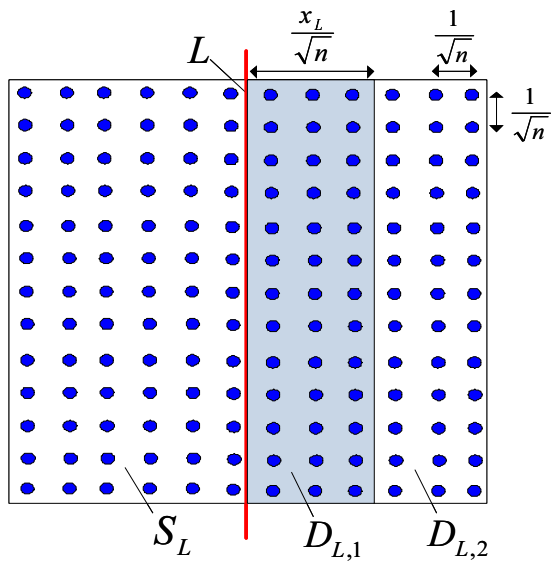}
  \caption{The cut $L$ in a two-dimensional dense regular network. $S_L$ and $D_L$ represent the sets of source and destination nodes, respectively, where $D_L$ is partitioned into two groups $D_{L,1}$ and $D_{L,2}$.}
  \label{FIG:cut_dense}
  \end{center}
\end{figure}

\begin{figure}[t!]
  \begin{center}
  \leavevmode \epsfxsize=0.56\textwidth   
  \leavevmode 
  \epsffile{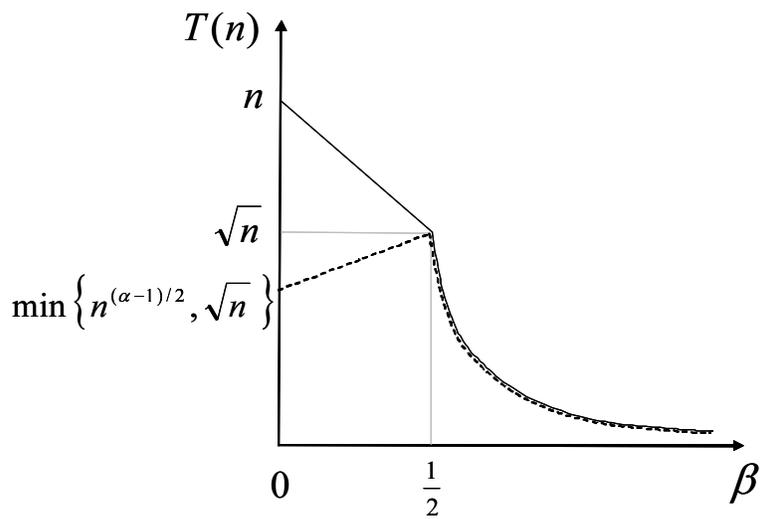}
  \caption{Upper (solid) and lower (dashed) bounds on the capacity scaling $T(n)$.}
  \label{FIG:bounds}
  \end{center}
\end{figure}

\begin{figure}[t!]
  \begin{center}
  \leavevmode \epsfxsize=0.55\textwidth   
  \leavevmode 
  \epsffile{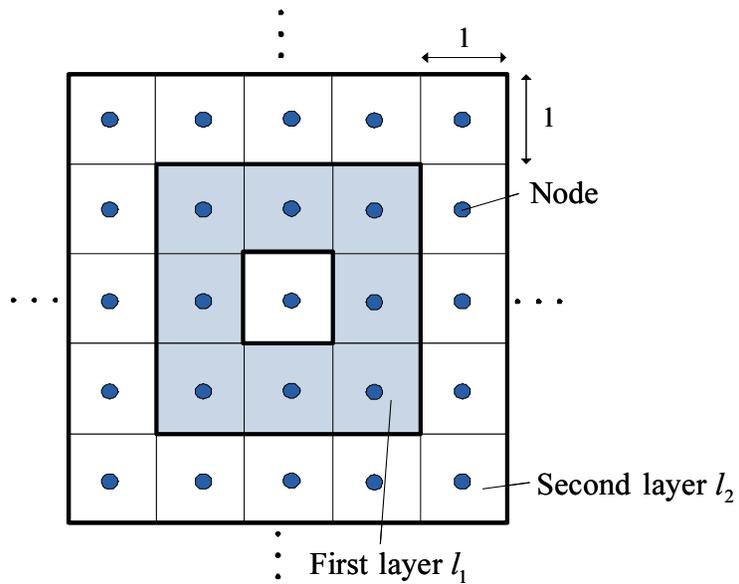}
  \caption{Grouping of interference routing cells in extended networks. The first layer $l_1$ represents the outer 8 shaded cells.}
  \label{FIG:layer}
  \end{center}
\end{figure}

\begin{figure}[t!]
  \begin{center}
  \leavevmode \epsfxsize=0.43\textwidth   
  \leavevmode 
  \epsffile{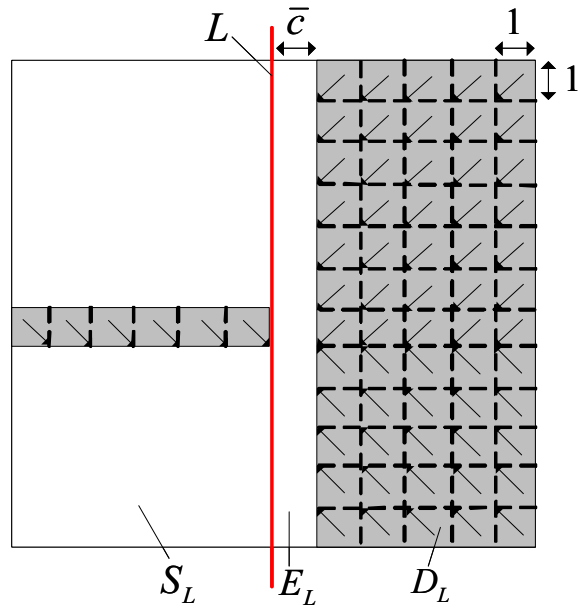}
  \caption{The node displacement to square vertices, indicated by arrows. The empty zone $E_L$ with width constant $\bar{c}$ is assumed for simplicity.}
  \label{FIG:displacement2}
  \end{center}
\end{figure}

\begin{figure}[t!]
  \begin{center}
  \leavevmode \epsfxsize=0.62\textwidth   
  \leavevmode 
  \epsffile{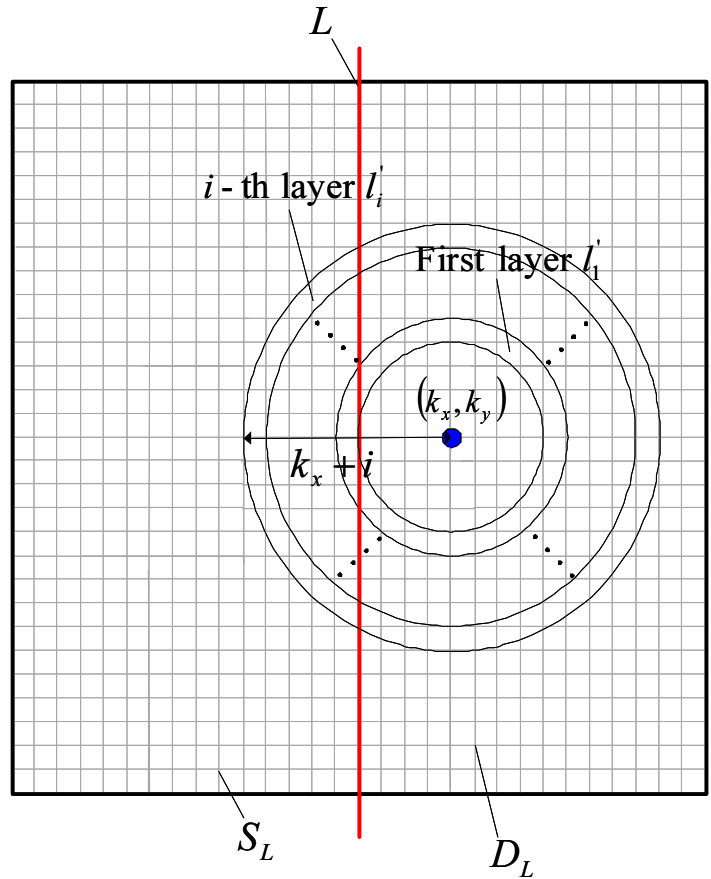}
  \caption{Grouping of source nodes in extended networks. There exist $\Theta(k_x)$ nodes in the first layer $l_1'$. This figure indicates the case where one destination is located at the position $(k_x,k_y)$. The source nodes are regularly placed at spacing 1 on the left half of the cut $L$.}
  \label{FIG:displacement}
  \end{center}
\end{figure}

\end{document}